\documentclass[aps,prx,twocolumn,superscriptaddress,nobibnotes]{revtex4-2}
\usepackage{graphicx}
\usepackage{amsmath}
\usepackage{amssymb}
\usepackage{amsthm}
\usepackage{braket}
\usepackage{siunitx}
\usepackage{color}
\usepackage[utf8]{inputenc}
\usepackage{here}
\usepackage{quantikz}
\usepackage[normalem]{ulem}

\usepackage{tikz}
\usepackage{tikz-cd}
\usetikzlibrary{arrows}
\usetikzlibrary{intersections}
\usetikzlibrary{shapes.geometric}
\usetikzlibrary{decorations.pathmorphing, patterns,shapes}
\usetikzlibrary{decorations.markings}

\usepackage[most]{tcolorbox}

\theoremstyle{plain}
\newtheorem{thm}{Theorem}
\newtheorem{lm}{Lemma}

\newtheorem*{thm*}{Theorem}
\theoremstyle{definition}
\newtheorem{dfn}{Definition}

\usepackage{physics}
\usepackage{mathtools}

\usepackage[hypertexnames=true]{hyperref}
\hypersetup{
    setpagesize=false,
    bookmarksnumbered=true,
    bookmarksopen=true,
    colorlinks=true,
    linkcolor=blue,
    citecolor=blue,
    urlcolor=blue,
    linktoc=all, 
    pdfusetitle, 
}

\makeatletter
\def\Hy@chapterbookmark#1{%
    \if@mainmatter
        \bookmark[rellevel=1,keeplevel,dest=\Hy@chapapp.\thechapter]{#1}%
    \fi
}
\makeatother

\begin{document}

\title{Singular value transformation for unknown quantum channels}

\author{Ryotaro Niwa}
\affiliation{Department of Physics, Graduate School of Science, The University of Tokyo, Hongo 7-3-1, Bunkyo-ku, Tokyo 113-0033, Japan}

\author{Zane Marius Rossi}
\affiliation{Department of Physics, Graduate School of Science, The University of Tokyo, Hongo 7-3-1, Bunkyo-ku, Tokyo 113-0033, Japan}

\author{Philip Taranto}
\affiliation{Department of Physics \& Astronomy, University of Manchester, Manchester M13 9PL, United Kingdom}

\author{Mio Murao}
\affiliation{Department of Physics, Graduate School of Science, The University of Tokyo, Hongo 7-3-1, Bunkyo-ku, Tokyo 113-0033, Japan}


\date{\today}

\begin{abstract}
Given the ability to apply an unknown quantum channel acting on a $d$-dimensional system, we develop a quantum algorithm for transforming its singular values. The spectrum of a quantum channel as a superoperator is naturally tied to its Liouville representation, which is in general non-Hermitian. Our key contribution is an approximate block-encoding scheme for this representation in a Hermitized form, given only \textit{black-box access} to the channel; this immediately allows us to apply polynomial transformations to the channel's singular values by quantum singular value transformation (QSVT). We then demonstrate an $O(d^3/\delta)$ upper bound and an $\Omega(d/\delta)$ lower bound for the query complexity of constructing a quantum channel that is $\delta$-close in diamond norm to a block-encoding of the unnormalized  Hermitized Liouville representation. We show our method applies practically to the problem of learning the $q$-th singular value moments of unknown quantum channels for arbitrary $q>2, q\in \mathbb{R}$, which has implications for testing if a quantum channel is entanglement breaking. Our results establish a general framework for manipulating black-box quantum channels with QSVT and related methods, enabling direct evaluation of the channel's spectral properties without tomography or significant classical post-processing.
\end{abstract}

\maketitle


\textit{Introduction.---}
Learning the properties of quantum systems is a central topic in quantum information science. In particular, numerous works have explored how to learn spectral properties of quantum states, as they characterize unitarily invariant properties of the system. Examples include representation theoretical methods for directly measuring the spectrum~\cite{PhysRevA.64.052311}, SWAP circuits for entanglement detection~\cite{PhysRevLett.94.040502, PhysRevLett.101.190503}, and classical post-processing methods for measuring entanglement entropy~\cite{Huang_2020, Elben_2022}.

Quantum singular value transformation (QSVT)~\cite{Gily_n_2019}, building on quantum signal processing (QSP)~\cite{lyc_16_equiangular_gates}, provides a unified framework underpinning various quantum algorithms including search, matrix inversion, and Hamiltonian simulation. QSVT enables one to apply polynomial transformations to the singular values of block-encoded matrices
(Def.~\ref{def:block_encoding}), and has recently been used to learn non-linear spectral properties of quantum states such as Rényi and Tsallis entropies~\cite{PhysRevA.104.022428, Liu_2025}. Extending similar capabilities to general unknown quantum channels remains a natural yet unexplored direction. 

One of the main obstacles in leveraging QSVT is the construction of an efficient block-encoding of the target operator, which can dominate the overall computational cost. When applying QSVT to states or channels, previous work has typically assumed access to unitaries that prepare purifications of quantum states~\cite{Gily_n_2019} or the Stinespring dilation of channels~\cite{PhysRevLett.128.220502}, respectively. Such assumptions, however, can be too strong for quantum learning: for an \textit{unknown} state or process, these purified forms are generally unavailable.

\begin{figure}[t]
\centering
\includegraphics[width=0.95\columnwidth]{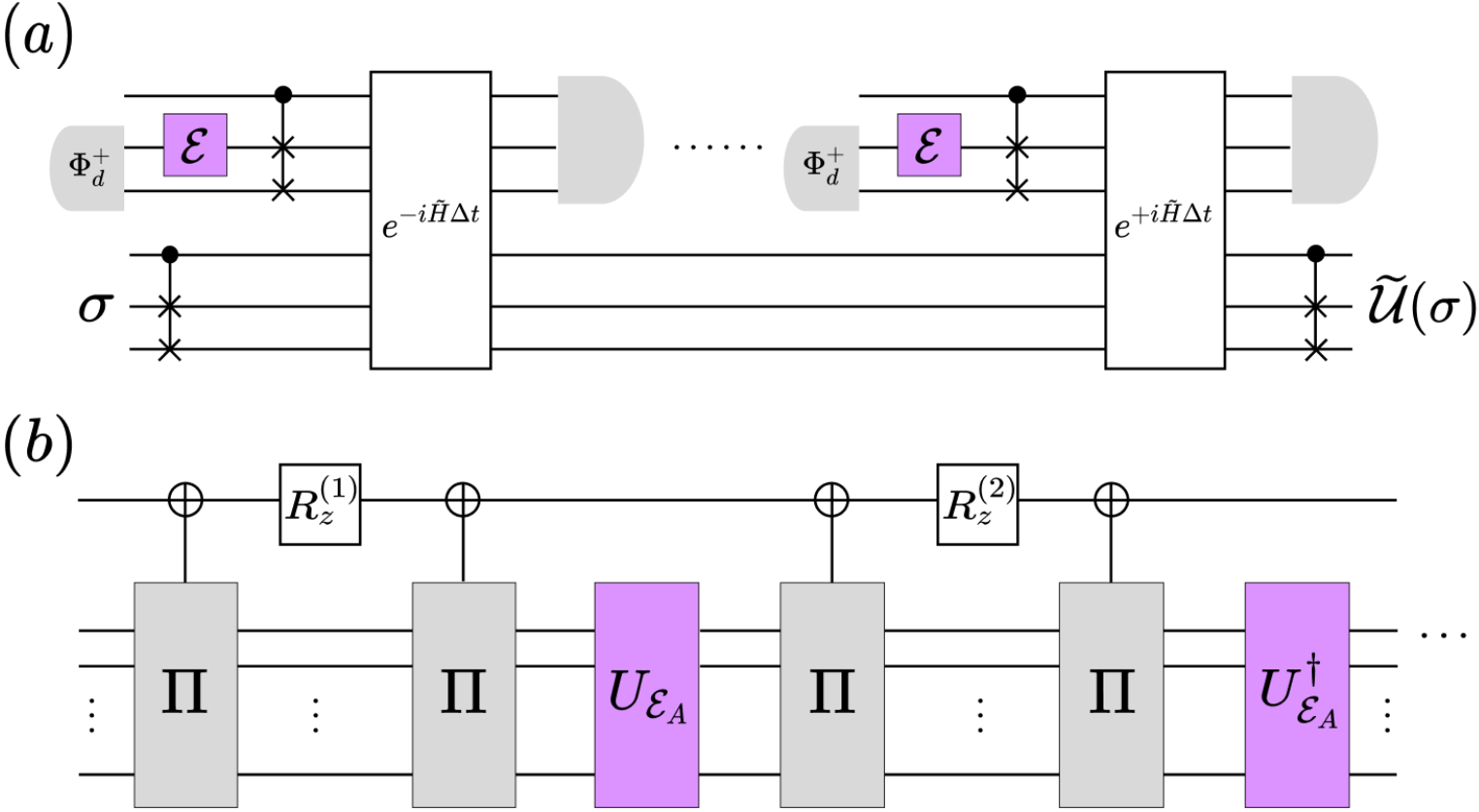}
\caption{(a) Construction of an approximate block-encoding of the Hermitized Liouville representation $\mathcal{E}_A$ of an unknown quantum channel $\mathcal{E}$. By coherently applying $\mathcal{E}$ multiple times, we implement a quantum channel $\widetilde{\mathcal{U}}$ that approximates a unitary block-encoding $U_{\mathcal{E}_A}$, obtained without access to Kraus operators or purifications. (b) Circuit for quantum singular value transformation (QSVT). $R_z^{(n)}=e^{-i\phi_n \sigma_z}$ are single qubit rotations, while grey boxes are projection-controlled-NOT operations locating $\mathcal{E}_A$ as a sub-block of $U_{\mathcal{E}_A}$.} 
\label{fig:channelQSVT}
\end{figure}

In this Letter, we introduce a novel framework for applying QSVT to quantum channels under the most natural access model for quantum process learning: one where the channel can be applied only as a black-box. As we aim to investigate the singular values of the quantum channel---a superoperator acting on density matrices---we first have to make sense of density matrices as vectors. Under such \textit{vectorization}, the action of a quantum channel is given by a matrix called the Liouville representation~\cite{Wallman_2014}, which naturally encodes the eigenvalues and singular values of the underlying channel. This representation is therefore well suited for studying spectral quantities such as the spectral gap governing the mixing time of Lindblad dynamics~\cite{Lindblad} or its impact on entanglement. Although the non-Hermiticity of the Liouville representation complicates our ability to process it directly on a quantum computer, we show how to construct---using only black-box access---a quantum process arbitrarily close to a block-encoding of its Hermitized form. With this block-encoding, standard tools such as QSVT are immediately applicable, allowing singular-value transformations on a quantum computer without prior knowledge of the channel's singular values or singular vectors (see FIG.~\ref{fig:channelQSVT}). Consequently, our method allows one to learn various non-linear spectral properties of arbitrary unknown quantum channels.

As a concrete application, we develop a protocol for estimating the $q$-th singular-value moment of an unknown quantum channel for arbitrary real $q>2$, based on a polynomial approximation to $f(x)=\frac{1}{2}x^{q-2}$. Combined with recent techniques for estimating quantum state moments $\Tr [\rho^q]$ for $q\in \mathbb{R}$~\cite{Liu_2025}, our approach extends and outperforms previous entanglement-detection methods~\cite{PhysRevLett.101.190503, PhysRevLett.129.260501}, which apply only to \textit{even integer} powers. This also shows that our scheme enables efficient tests for entanglement breaking channels based on the reshuffling criterion of entanglement~\cite{Rudolph_2005,chen2003matrixrealignmentmethodrecognizing}.

Our results establish a general framework for manipulating black-box quantum channels with QSVT and related methods, enabling direct manipulation of channel spectra without tomography or significant classical post-processing.


\textit{Preliminaries.---} Let $\mathfrak{S}(\mathcal{H}), \mathfrak{L}(\mathcal{H})$ respectively denote the set of density matrices and linear operators on the Hilbert space $\mathcal{H}$, and let $\mathfrak{C}(\mathcal{H})$ denote the set of completely positive trace-preserving (CPTP) maps acting on $\mathfrak{L}(\mathcal{H})$. A quantum channel $\mathcal{E}\in \mathfrak{C}(\mathcal{H})$ has various representations; two particularly important ones relevant to our manuscript are the \textit{Liouville representation} and the \textit{Choi representation}~\footnote{The Liouville representation $\mathcal{E}_
A$ and the Choi state $\mathcal{E}_B$ are sometimes called the {\it A-form} and the {\it B-form} matrices, respectively (as they were originally derived in Ref.~\cite{PhysRev.121.920})}, as we describe below.

\textbf{(A)} Liouville representation: When $\rho \in \mathfrak{S}(\mathcal{H})$ is represented by its vectorized form $|\rho\rangle \rangle \in \mathcal{H}\otimes \mathcal{H}$ as 
\begin{align}
    \rho = \sum_{ij} \rho_{ij} |i\rangle \langle j| \rightarrow |\rho\rangle \rangle := \sum_{ij} \rho_{ij}|i\rangle \otimes |j\rangle,
\end{align} 
the action of a quantum channel $\mathcal{E} \in \mathfrak{C}(\mathcal{H})$ becomes the \textit{matrix multiplication}
\begin{align}\label{eq:Liouvillematrix}
    |\mathcal{E}(\rho)\rangle \rangle = \mathcal{E}_A |\rho\rangle \rangle. 
\end{align}
The matrix $\mathcal{E}_A \in \mathfrak{L}(\mathcal{H} \otimes \mathcal{H})$ is called the \textit{Liouville representation}~\cite{Wallman_2014} of the channel $\mathcal{E}$, which is in general \textit{non-Hermitian}. This suggests that $\mathcal{E}_A$ itself cannot be obtained as a quantum state. On the other hand, an important consequence of Eq.~\eqref{eq:Liouvillematrix} is that $\mathcal{E}_A$ is naturally associated with the spectrum of the channel $\mathcal{E}$ as a superoperator. Indeed, if 
\begin{align}
    \mathcal{E}(\rho_n) = \lambda_n \rho_n \quad (\lambda_n\in \mathbb{C})
\end{align}
for some $\rho_n \in \mathfrak{L}(\mathcal{H})$ (not necessarily $\rho_n \in \mathfrak{S}(\mathcal{H})$), then
\begin{align}
    \mathcal{E}_A|\rho_n \rangle\rangle = |\mathcal{E}(\rho_n) \rangle \rangle = \lambda_n |\rho_n \rangle \rangle. 
\end{align}
Thus the eigenvalues of the non-Hermitian matrix $\mathcal{E}_A$ coincide with the eigenvalues of the quantum channel (see Supplemental Material A, which includes Refs.~\cite{Watrous_2018, Wallman_2014}). Analogously, the singular values of the quantum channel $\mathcal{E}$ can be defined as the square root of the eigenvalues of the map $\mathcal{E}^\dagger \circ \mathcal{E}$. Since map composition becomes matrix multiplication in the Liouville representation due to~\eqref{eq:Liouvillematrix}, the \textit{singular values} of the quantum channel coincide with those of $\mathcal{E}_A$:
\begin{align}
    \mathcal{E}_A = U\Sigma V^\dagger, \Sigma=\mathrm{diag}[\sigma_1,\sigma_2, \cdots, \sigma_{d^2}], 
\end{align}
where $U,V \in \mathfrak{L}(\mathcal{H} \otimes \mathcal{H})$ are unitary matrices. Note that the eigenvalues and singular values of $\mathcal{E}_A$ are not equivalent in general because of the non-Hermiticity of $\mathcal{E}_A$. 

\textbf{(B)} Choi representation: Alternatively, under the Choi–Jamio{\l}kowski isomorphism~\cite{JAMIOLKOWSKI1972275, CHOI1975285}, a quantum channel $\mathcal{E}$ can be identified with its (normalized) \textit{Choi state} $\mathcal{E}_B \in \mathfrak{L}(\mathcal{H} \otimes \mathcal{H})$ defined by
\begin{align}
    \mathcal{E}_B := (\mathcal{E}\otimes I) |\Phi_d^+\rangle \langle \Phi_d^+|,
\end{align}
where $|\Phi_d^+\rangle:=\frac{1}{\sqrt{d}} \sum_i |i\rangle \otimes |i\rangle$ denotes the maximally entangled state (MES), and $d=\dim \mathcal{H}$~\footnote{We note that the (unnormalized) Choi matrix $d\mathcal{E}_B$ is also often considered in the literature, which we distinguish in this manuscript from the Choi state $\mathcal{E}_B$.}. It reflects some foundational qualities of the map, for instance, any completely positive map corresponds to a positive semidefinite $\mathcal{E}_B$, with trace preservation further implying an affine constraint. Therefore, for any $\mathcal{E} \in \mathfrak{C}(\mathcal{H})$, the Choi state $\mathcal{E}_B$ can be prepared as a quantum state, i.e., $\mathcal{E}_B\in \mathfrak{S}(\mathcal{H}\otimes \mathcal{H})$, in contradistinction to the Liouville representation $\mathcal{E}_A$. On the other hand, map composition under the Choi representation translates to the \textit{link product}~\cite{PhysRevA.80.022339} rather than matrix multiplication, and the spectrum of $\mathcal{E}_B$ is completely different from that of the quantum channel $\mathcal{E}$. Thus, the Choi representation is not necessarily suited for discussing the channel's spectral properties, which are naturally encapsulated in the Liouville representation $\mathcal{E}_A$. 

It is known that the Liouville representation $\mathcal{E}_A$ and the Choi state $\mathcal{E}_B$ of a channel $\mathcal{E}$ are related to each other via the \textit{reshuffling operation} $\mathcal{R}:|i \rangle \langle \underline{j}| \otimes |\underline{k} \rangle \langle l| \rightarrow |i \rangle \langle \underline{k}| \otimes |\underline{j} \rangle \langle l|$ (see, e.g., Ref.~\cite{Milz_2017}):
\begin{align}\label{eq:Choi}
    d\mathcal{R}(\mathcal{E}_B) &= \mathcal{E}_A.
\end{align}
Since $\mathcal{R}$ is linear but not completely-positive, it cannot be applied to the Choi state $\mathcal{E}_B$ as a quantum operation to produce $\frac{1}{d}\mathcal{E}_A$; this makes it challenging at first glance to construct a desired block-encoding of $\mathcal{E}_A$. 


\textit{Block-encoding the Liouville representation.---} To address the issue of non-Hermiticity, we construct an approximate block-encoding for the Liouville representation in a \textit{Hermitized} form.  
\begin{dfn}\label{def:block_encoding}
    If $A$ is an $n$-qubit operator and $U$ is an $(n+a)$ qubit unitary operator that satisfies
    \begin{align}
        \lVert A-\alpha(\bra{0}^{\otimes a}\otimes I) U (\ket{0}^{\otimes a}\otimes I) \rVert \leq \epsilon,
    \end{align}
    we say that $U$ is an $(\alpha, a, \epsilon)$\textit{-block encoding} of $A$. 
\end{dfn}

\begin{thm}\label{thm1}
    Given the ability to apply an unknown quantum channel $\mathcal{E}$ acting on a $d$-dimensional system, one can construct a quantum channel $\widetilde{\mathcal{U}}$ that is $\delta$-close in diamond norm to a $(1, 3, 0)$-block encoding of $\frac{2}{\pi}H$ with $O(\frac{d^{2k+1}}{\delta} \log^2 \frac{1}{\delta})$ queries to $\mathcal{E}$, where $H$ is the Hermitized Liouville representation 
    \begin{align}
        H:= \frac{1}{2d^{1-k}}\mqty[ O & \mathcal{E}_A\\ \mathcal{E}_A^\dagger & O]. 
    \end{align}
    Here, $O$ inside $H$ denotes a zero matrix. $k\in [0,\frac{1}{2}]$ ensures for arbitrary channels that $\lVert H \rVert_\infty \leq \frac{1}{2}$. Restricted to unital channels, this condition relaxes to $k\in [0,1]$. 
\end{thm}
\begin{proof}
Consider the states $\rho_\pm \in \mathfrak{S}(\mathcal{H}_X \otimes \mathcal{H}_Y \otimes \mathcal{H}_Z)$ depicted below, where $\mathcal{H}_X =\mathbb{C}^2, \mathcal{H}_Y =\mathbb{C}^d$ and $\mathcal{H}_Z =\mathbb{C}^d$: 
\begin{figure}[H]
\centering
\includegraphics[width=0.45\columnwidth]{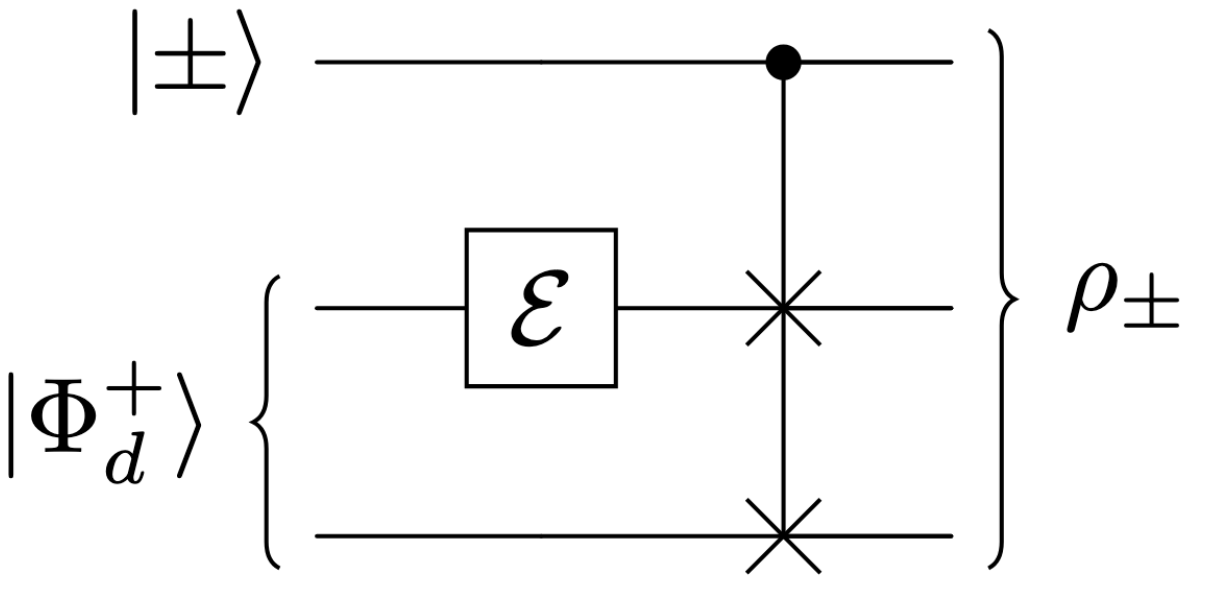}
\label{fig:ctrl_swap}
\end{figure}
\noindent 
The top line represents the control qubit system $\mathcal{H}_X$, while the second and third lines represent $d$-dimensional systems  $\mathcal{H}_Y$ and $\mathcal{H}_Z$, respectively. They are obtained by applying the controlled-\texttt{SWAP} operation to the Choi state $\mathcal{E}_B$ with the control qubit initialized as $\ket{\pm} := (\ket{0} \pm \ket{1})/\sqrt{2}$. The states $\rho_\pm$ have the form 
\begin{align}\label{eq:AAdaggerstate}
    \rho_\pm &= \frac{1}{2}|0\rangle \langle 0| \otimes (\mathcal{E}\otimes I)|\Phi_d^+\rangle \langle \Phi_d^+| \pm \frac{1}{2d}|0\rangle \langle 1| \otimes (\mathcal{E}_A\mathbb{F}_{YZ})^{T_Z}\nonumber\\
    &\pm \frac{1}{2d}|1\rangle \langle 0| \otimes (\mathbb{F}_{YZ}\mathcal{E}_A^\dagger)^{T_Z} +\frac{1}{2}|1\rangle \langle 1| \otimes (I \otimes \mathcal{E})|\Phi_d^+\rangle \langle \Phi_d^+|, 
\end{align}
where $\mathbb{F}_{YZ}$ denotes the \texttt{SWAP} operator acting on $\mathcal{H}_Y$ and $\mathcal{H}_Z$, and $(\cdot)^{T_Z}$ denotes the partial transposition with respect to $\mathcal{H}_Z$. 

Utilizing a variant of density matrix exponentiation~\cite{Lloyd_2014, Wei:2024znk}, the unitary $e^{-i\rho^{T_Z}t}$ for a state $\rho \in \mathfrak{S}(\mathcal{H}_X \otimes \mathcal{H}_Y \otimes \mathcal{H}_Z)$ can be  approximately applied on a state $\sigma \in \mathfrak{S}(\mathcal{H}_{X'} \otimes \mathcal{H}_{Y'} \otimes \mathcal{H}_{Z'})$ by interspersing $e^{-i\widetilde{H} \Delta t}$ with $\widetilde{H}:=\mathbb{F}_{XX'\cup YY'} \otimes d|\Phi_d^+\rangle \langle \Phi_d^+|_{ZZ'}$ and tracing out the Hilbert space of $\rho$, which gives
\begin{align}
  \Tr_{XYZ}[e^{-i\widetilde{H}\Delta t} &(\rho\otimes \sigma) e^{i\widetilde{H}\Delta t}] \nonumber\\
  =& \sigma -i\Delta t[\rho^{T_Z},  \sigma]+O(d\Delta t^2)\nonumber\\
  =&  e^{-i\rho^{T_Z} \Delta t} \sigma e^{i\rho^{T_Z} \Delta t} +O(d\Delta t^2).  
\end{align}
Here, the \texttt{SWAP} operator $\mathbb{F}_{XX'\cup YY'}$ acts on $\mathcal{H}_X \otimes \mathcal{H}_Y$ of $\rho_\pm$ and $\mathcal{H}_{X'} \otimes \mathcal{H}_{Y'}$ of $\sigma$, while the unnormalized MES $d|\Phi_d^+\rangle \langle \Phi_d^+|_{ZZ'}$ acts on $\mathcal{H}_Z$ of $\rho$ and $\mathcal{H}_{Z'}$ of $\sigma$. (See Supplemental Material B.1, which includes Refs.~\cite{Lloyd_2014, Wei:2024znk} for details). Choosing  $\rho=\rho_\pm$, we can approximate the unitary 
\begin{align}\label{eq:UEAprime}
    U_{\mathcal{E}_A}' & := e^{i\rho_-^{T_Z}\Delta t} e^{-i\rho_+^{T_Z} \Delta t} \nonumber\\
    &= \mathrm{exp}\qty[-\frac{i}{d}\mqty(O & \mathcal{E}_A\mathbb{F}_{YZ}
    \\ \mathbb{F}_{YZ}\mathcal{E}_A^\dagger & O)\Delta t]+O(d\Delta t^2). 
\end{align}
Further choosing $\Delta t = \frac{2\delta}{d^{k+1}}$ and applying the above procedure $O(\frac{d^{2k+1}}{4\delta})$ times, we have $O(d\Delta t^2) \times O(\frac{d^{2k+1}}{4\delta}) = O(\delta)$. Thus, we arrive at a $\delta$-close approximation of
\begin{align}
    \widetilde{U}_{\mathcal{E}_A} := \mathrm{exp}\qty[-\frac{i}{2d^{1-k}}\mqty(O & \mathcal{E}_A\mathbb{F}_{YZ}
    \\ \mathbb{F}_{YZ}\mathcal{E}_A^\dagger & O)].
\end{align}
By conjugating $\widetilde{U}_{\mathcal{E}_A}$ with controlled-\texttt{SWAP}s, we obtain 
\begin{align}
    U_{\mathcal{E}_A} &= \mqty(I & O\\ O & \mathbb{F}_{YZ})U'_{\mathcal{E}_A}\mqty(I & O\\ O & \mathbb{F}_{YZ}) \nonumber\\
    &= \mathrm{exp}\qty[-\frac{i}{2d^{1-k}}\mqty(O & \mathcal{E}_A
    \\ \mathcal{E}_A^\dagger & O)]:=e^{-iH}. 
\end{align}
Here, $H$ satisfies $\lVert H\rVert_\infty \leq \frac{1}{2}$ as long as $k\leq \frac{1}{2}$ for general channels and $k\leq 1$ for unital channels, respectively (see Supplemental Material A, Lemma 3). Note that $U_{\mathcal{E}_A}^\dagger$ is also obtained by replacing $e^{i\rho_-^{T_Z}\Delta t}e^{-i\rho_+^{T_Z}\Delta t}\to e^{-i\rho_-^{T_Z}\Delta t}e^{i\rho_+^{T_Z}\Delta t}$ in Eq.~\eqref{eq:UEAprime}. Furthermore, control-$U_{\mathcal{E}_A}$ and its inverse can be similarly obtained by replacing $\rho_\pm$ with $|1\rangle \langle 1| \otimes \rho_\pm$~\cite{Kimmel_2017, Wei:2024znk}.

Now, suppose $U = e^{-iH}$, where $\lVert H \rVert_\infty \leq \frac{1}{2}$. Given access to control-$U$ and its inverse, we can obtain a block encoding of $\sin H$ since (c-$U$ indicating controlled-$U$)
\begin{align}
    \sin H =  (\langle +| \otimes I) (\text{c-}U) (Y \otimes I) (\text{c-}U^\dagger) (|+\rangle \otimes I).    
\end{align}
For $\epsilon' \in (0, \frac{1}{2}]$, there exists an efficiently computable odd polynomial $P(x)\in \mathbb{R}[x]$ of degree $O(\log \frac{1}{\epsilon'})$ such that $\lVert P(x)\rVert_{[-1,1]} \leq 1$ and $\lVert P(x)-\frac{2}{\pi}\arcsin{x}\rVert_{[-\frac{1}{2}, \frac{1}{2}]} \leq \epsilon'$~(Ref. \cite{Gily_n_2019}, Lemma 70). Thus, we can apply the quantum singular value transformation (QSVT)~\cite{Gily_n_2019} (see Supplemental Material B.2, which includes Refs.~\cite{Gily_n_2019, mrtc_unification_21, lyc_16_equiangular_gates, lc_17_simulation, cs_qsvt_tang_tian, Liu_2025}) based on the polynomial approximation of $\frac{2}{\pi} \arcsin(x)$ on the domain $[-\frac{1}{2}, \frac{1}{2}]$ to obtain a block-encoding of $\frac{2}{\pi}H$ with $O(\log \frac{1}{\epsilon'})$ queries to $U$.

Putting this all together, we can construct a quantum channel $\widetilde{\mathcal{U}}$ that is $\delta$-close in diamond norm to a $(1, 2, \epsilon')$-block encoding of $\frac{2}{\pi} H$ using $O\qty(\frac{d^{2k+1}}{\delta} \log^2 \frac{1}{\epsilon'})$ queries to the black-box channel $\mathcal{E}$ (see FIG.~\ref{fig:channelQSVT} (a)). The substitution $\epsilon' \rightarrow O(\delta), \delta\rightarrow \delta/2$ in the aforementioned statement further implies that we have a quantum channel $\widetilde{\mathcal{U}}$ that is $\delta$-close in diamond norm to a $(1, 3, 0)$-block-encoding of $\frac{2}{\pi} H$, using $O\qty(\frac{d^{2k+1}}{\delta} \log^2 \frac{1}{\delta})$ queries to $\mathcal{E}$. 
\end{proof}

We note that it is also possible to first implement the reshuffling operation \textit{at the state level} in a Hermitized form, using an auxiliary MES. This allows us to estimate the full singular value spectrum of unital quantum channels. However, when similarly used for constructing an approximate block-encoding of the Hermitized Liouville representation, it has a sub-optimal normalization factor (see Supplemental Material C, which includes Refs.~\cite{Kimmel_2017, PhysRevA.64.052311, odonnell2016efficientquantumtomographyii}). 


\textit{Lower bound for query complexity.---} 
To evaluate the efficiency of our proposed protocol, we now prove a corresponding lower bound for the query complexity of constructing an approximate block-encoding for the channel's Hermitized Liouville representation, taking $d=2^n$. 
\begin{thm}\label{thm2} 
Given the ability to apply an unknown quantum channel $\mathcal{E}$ acting on a $d=2^n$ dimensional system, a universal scheme that converts $\mathcal{E}$ into a channel $\widetilde{\mathcal{U}}$ that is $\delta$-close in diamond norm to the unitary channel 
\begin{align}
\mathcal{U}(t)(\cdot)&=U_{\mathcal{E}_A}(t)(\cdot)U_{\mathcal{E}_A}(t)^\dagger\nonumber\\
U_{\mathcal{E}_A}(t) &= \mathrm{exp} \qty(-i\mqty[O & \mathcal{E}_A\\ \mathcal{E}_A^\dagger & O]t)
\end{align}
requires $\Omega(\frac{dt^2}{\delta})$ applications of the unknown channel $\mathcal{E}$. 
\end{thm}
\noindent Our proof strategy is to consider a specific channel discrimination task (see Supplemental Material D, which includes Refs.~\cite{Chiribella_2008, Gour_2019}). Since Hamiltonian simulation algorithms via QSVT~\cite{Low:2016znh} have dimension-independent query complexity that is linear in $t$, the above theorem implies a $\Omega(\frac{dt^2}{\delta})$ query complexity lower bound for approximating a block-encoding unitary of the Hermitized Liouville representation. We can compare this result with Thm.~\ref{thm1}, which implies an $\widetilde{O}(\frac{d^3t^2}{\delta})$ upper bound for approximating $\mathcal{U}(t)$ up to precision $\delta$. 


\textit{QSVT for unknown quantum channels.---}
With this approximate block-encoding of the Liouville representation in Thm.~\ref{thm1} at hand, we can freely apply QSVT to implement polynomial transformations to the singular values of unknown quantum channels (see FIG.~\ref{fig:channelQSVT} (b)). 
\begin{thm}\label{cor2} 
    Let $P_{\epsilon''}(x)$ be an $\epsilon''$-approximating polynomial for the function $f(x)$ on $[-1,1]$, with $\deg P = Q(\epsilon'')$. A quantum channel $\delta$-close in diamond norm to a $(1,6,0)$-block-encoding of $f(H)$, where 
    \begin{align*}
    H = \frac{1}{d^{1-k}} \mqty[O & \mathcal{E}_A\\ \mathcal{E}_A^\dagger & O]
    \end{align*}
    can be constructed using $O\qty(\frac{d^{2k+1}Q^2}{\delta} \log^2 \frac{Q}{\delta})$ queries to the unknown channel $\mathcal{E}$ with $k \in [0,\frac{1}{2}]$ for general channels. Restricted to unital channels, this condition relaxes to $k\in [0,1]$. 
\end{thm}
This result is a direct consequence of Thm.~\ref{thm1} and the notion of a \textit{samplizer}~\cite{wang_et_al:LIPIcs.ESA.2024.101}, which establishes the cost inherent in transforming a circuit querying unitaries block-encoding a density matrix $\rho$ into a quantum circuit given access to copies of $\rho$ itself (see Supplemental Material B.3, which includes Refs.~\cite{Lloyd_2014, wang_et_al:LIPIcs.ESA.2024.101}). Note that the normalization factor inside the matrix function $f$ vanishes for unital channels by setting $k=1$.


\textit{Learning singular value moments.---} As a concrete example, we consider the problem of learning the $q$-th \textit{singular value moment} of an unknown quantum channel $\mathcal{E}$:
\begin{align}
    S_q := \frac{1}{d^q}\sum_{i=1}^{d^2} \sigma_i^q = \Tr [\mathcal{R}(\mathcal{E}_B)\mathcal{R}(\mathcal{E}_B)^\dagger]^\frac{q}{2}. 
\end{align}
The first moment $S_1$ serves as a criterion for testing if the channel is entanglement breaking. This is because the reshuffling criterion of entanglement~\cite{Rudolph_2005,chen2003matrixrealignmentmethodrecognizing} states that 
\begin{align}
    \rho \in \mathrm{SEP} \implies \lVert \mathcal{R}(\rho)\rVert_1 \leq 1, 
\end{align}
where $\mathrm{SEP}$ is the set of separable quantum states, and a quantum channel is entanglement breaking if and only if its Choi state is separable~\cite{Horodecki_2003}. However, directly estimating $S_1$ can generally be resource-intensive (see Supplemental Material E, which includes Refs.~\cite{Haah_2017}). Thus, previous works in entanglement detection have employed various methods such as \texttt{SWAP} circuits~\cite{PhysRevLett.101.190503} or classical shadow tomography~\cite{PhysRevLett.129.260501} to estimate even integer moments $S_{2n} (n\in \mathbb{N})$, which serve as a lower bound for $S_1$ in the regime $n>1$. 

Here we demonstrate how our block-encoding scheme not only allows us to learn $S_q$ for arbitrary \textit{real} moments $q>2, q\in \mathbb{R}$, but also to exponentially improve query complexity in certain regimes. Our overall scheme is to perform QSVT on the block-encoded Hermitized Liouville representation based on an $\epsilon''$-approximating even polynomial for $f(x)=\frac{1}{2}x^{q-2}$, which uniformly converges in the range $[-1,1]$ and is of degree $O(1/{\epsilon''^{\frac{1}{q-2}}})$. Then, we apply the circuit depicted below: 
\begin{figure}[H]
\centering
\includegraphics[width=1.0\columnwidth]{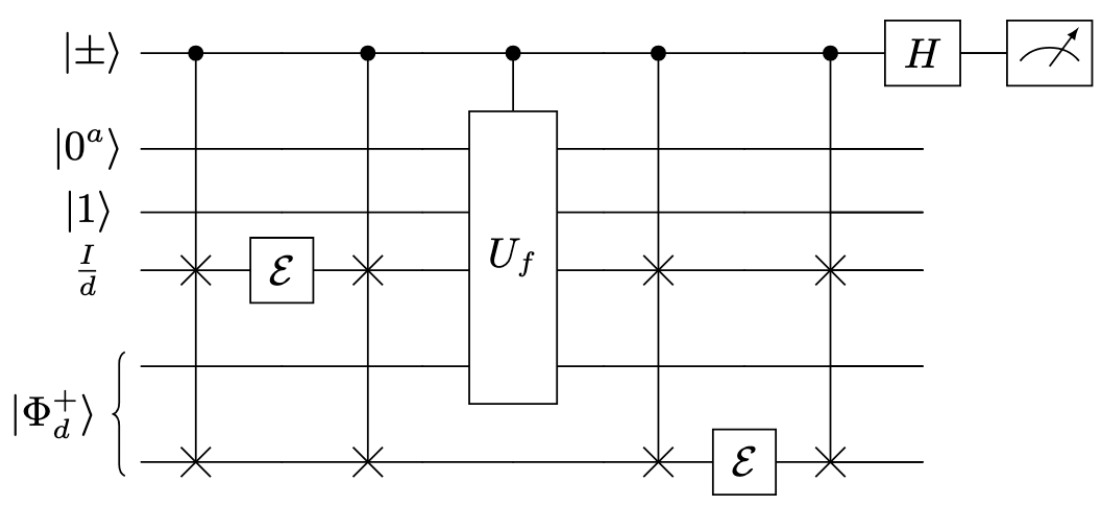}
\label{fig:modifiedHadamard}
\end{figure}
\noindent Here, $U_f$ is the (approximate) unitary block-encoding of $M \simeq f(H)$ obtained via QSVT, where $H$ is the (normalized) Hermitized Liouville representation. Repeating this circuit produces an estimate for $\frac{1}{d^2}\Tr[\mathcal{E}_A^\dagger \mathcal{E}_A \langle 1|M|1\rangle]$. See Supplemental Material B.4 for comparison with the standard Hadamard test~\cite{PhysRevLett.81.5672}. The overall query complexity for estimating $S_q$ to additive precision $\epsilon$ with probability larger than $1-\delta$ is 
\begin{align}\label{eq:querycomp}
\begin{dcases}
    \widetilde{O}\qty(\frac{\log \frac{1}{\delta}}{d^{\frac{3}{2}q-4}\epsilon^{3+\frac{2}{q-2}}}) \quad(\textrm{general channels})\\
    \widetilde{O}\qty(\frac{\log \frac{1}{\delta}}{d^{3q-7}\epsilon^{3+\frac{2}{q-2}}}) \quad(\textrm{unital channels})\\
\end{dcases}
\end{align}
where $\widetilde{O}$ ignores logarithmic factors in $d,\epsilon$ (see Supplemental Material F for details). When the unknown channel is close to the ``pure regime", where $S_q \simeq d^{2-q}$, this reduces the number of channel uses exponentially for $q>8$, compared to existing methods such as the \texttt{SWAP} circuit~\cite{PhysRevLett.101.190503} (see Supplementary Material G for details). It also outperforms local classical shadow tomography  performed on the channel's Choi states~\cite{PhysRevLett.129.260501}, which was originally developed specifically for $q=4$. If the channel is further promised to be unital, our method outperforms both local and global classical shadow tomography developed for $S_4$, and outperforms SWAP circuits in the region $q>3.5$. Notice how our black-box channel access model allows us to utilize the time degree of freedom in the above circuit, which was absent in states.

\textit{Conclusion.---} We have developed a block-encoding scheme for the Liouville representation of a quantum channel under the most  natural assumption for quantum process learning: the black-box channel access model. Our scheme achieves a $\delta$-approximate block-encoding of the unnormalized Hermitized Liouville representation with $O(d^3/\delta)$ queries. We have also provided a corresponding $\Omega(d/\delta)$ lower bound based on channel discrimination arguments.  These results clarify the fundamental cost of accessing the channel's spectrum as a block-encoded matrix. As a practical application of our method, we considered the problem of calculating arbitrary singular value moments for $q>2, q \in \mathbb{R}$,  which has implications for testing if a given black-box channel is entanglement-breaking. Our method shows an advantage in query complexity compared with existing methods that can be performed on Choi states. 

There are many directions for future research. First, we note that the upper and lower bounds proven in this Letter do not match---improving these bounds is of fundamental interest. From a more practical viewpoint, investigating the role of locality or structure of the black-box channels could permit more efficient protocols for reasonable and common experimental settings.
We also remark that it is possible to consider different access models such as the block-encoding of Kraus operators or the Stinespring dilation of the channel. While such access models are not necessarily suited for quantum process learning, we expect that such  ``purified" access models could result in a more efficient protocol for transforming the spectrum of the Liouville representation. We also anticipate that the recently proposed framework of quantum eigenvalue transformation \cite{Low_2024} could be combined with the block-encoding of the Liouville representation to efficiently transform the eigenvalues of quantum channels with \textit{complex} spectra. By providing a block-encoding circuit for the Liouville representation, our work opens new avenues for  processing quantum channels with established techniques such as QSVT.


\textit{Acknowledgments.---} We are grateful to Pawel Wocjan,
Anirban Chowdhury, Toshinari Itoko, and Fuchuan Wei for
discussions and comments. This work was supported by
MEXT Quantum Leap Flagship Program (MEXT QLEAP)
JPMXS0118069605 and JPMXS0120351339; Japan Science
and Technology Agency (JST) as part of Adopting Sustainable
Partnerships for Innovative Research Ecosystem
(ASPIRE), Grant No. JPMJAP25A3; JST CREST, Grant No.
JPMJCR25I5; JST NEXUS, Grant No. JPMJNX26C9; JSPS
KAKENHI Grant No. 23K21643; and IBM Quantum. Z.M.R.
acknowledges funding from the Japan Society for the Promotion
of Science (JSPS) Postdoctoral Fellowship for Research
in Japan (KAKENHI 24KF0136).

\bibliography{references}

\newpage
\onecolumngrid
\appendix
\newpage 

\section*{Supplemental Material}
\setcounter{section}{0}

\section*{Notations}
\subsection{Norms}
Here we clarify the definition for the norms. The \textit{Schatten $p$-norm} of the square matrix $A$ is 
\begin{equation}
    \lVert A \rVert_p := \qty(\sum_i \sigma_i^p)^{\frac{1}{p}},
\end{equation}
where $\sigma_i$ represents the singular values of $A$. It satisfies the monotonicity 
\begin{equation}
    \lVert A\rVert_{\infty} \leq \cdots \leq \lVert A\rVert_{2} \leq \lVert A\rVert_{1}.
\end{equation}
$||\cdot ||_1$ is called the \textit{trace norm} and $||\cdot ||_{\infty}$ is called the \textit{spectral norm}. 
The Hölder's inequality 
\begin{equation}
    \lVert AB\rVert_{1} \leq \lVert A\rVert_{p}\lVert B\rVert_q
\end{equation}
holds for $p,q\in [1, \infty]$ satisfying $\frac{1}{p}+\frac{1}{q} = 1$. The \textit{diamond norm distance} between two quantum channels $\mathcal{E}_1, \mathcal{E}_2$ acting on system $A$ can be defined as 
\begin{align}
    \lVert \mathcal{E}_1-\mathcal{E}_2 \rVert_\diamond = \frac{1}{2}\sup_{B}\max_{\rho_{AB}} \lVert (\mathcal{E}_1 \otimes I_B)(\rho_{AB})-(\mathcal{E}_2\otimes I_B)(\rho_{AB}) \rVert_1, 
\end{align}
where $B$ is an ancilla system and $\rho_{AB}$ is a state on systems $A \cup B$. 

\subsection{Diagram notations}\label{app::diagram} 
Here we introduce diagram notations that facilitate calculations. We assign an index to each bra and ket degrees of freedom, and represent the Kronecker delta $\delta_{ab}$ as a line connecting the indices $a$ and $b$. This leads to us to denote the MES $|\Phi_d^+\rangle\langle \Phi_d^+|$ and the \texttt{SWAP} operator $\mathbb{F}_{AB}$ in the following way; 
\begin{align}
\begin{dcases}
    |\Phi_d^+\rangle\langle \Phi_d^+| =  \frac{1}{d} \sum_{i,j} |i\rangle \langle j| \otimes |i\rangle \langle j| = \frac{1}{d} \quad \vcenter{\hbox{\includegraphics[clip,scale=0.2]{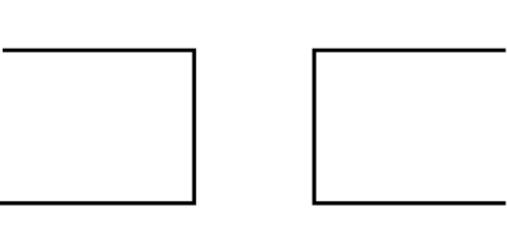}}}\\
    \mathbb{F}_{AB} = \sum_{i,j} |i\rangle \langle j| \otimes |j\rangle \langle i| = \vcenter{\hbox{\includegraphics[clip,scale=0.15]{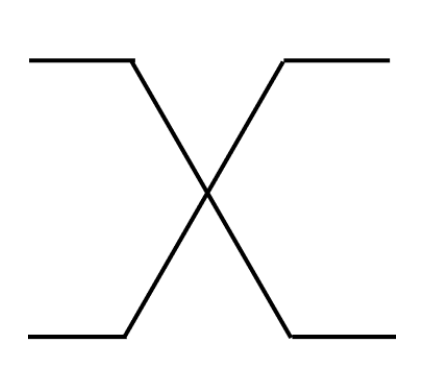}}}.
\end{dcases}
\end{align}
It is then easier to see relations such as 
\begin{align}
    d(|\Phi_d^+\rangle\langle \Phi_d^+|)^{T_A} = \, \mathbb{F}_{AB},
\end{align}
where $(\cdot)^{T_A}$ denotes the partial transposition over system A. We also denote the Kraus operators as below; 
\begin{align}
    \vcenter{\hbox{\includegraphics[clip,scale=0.12]{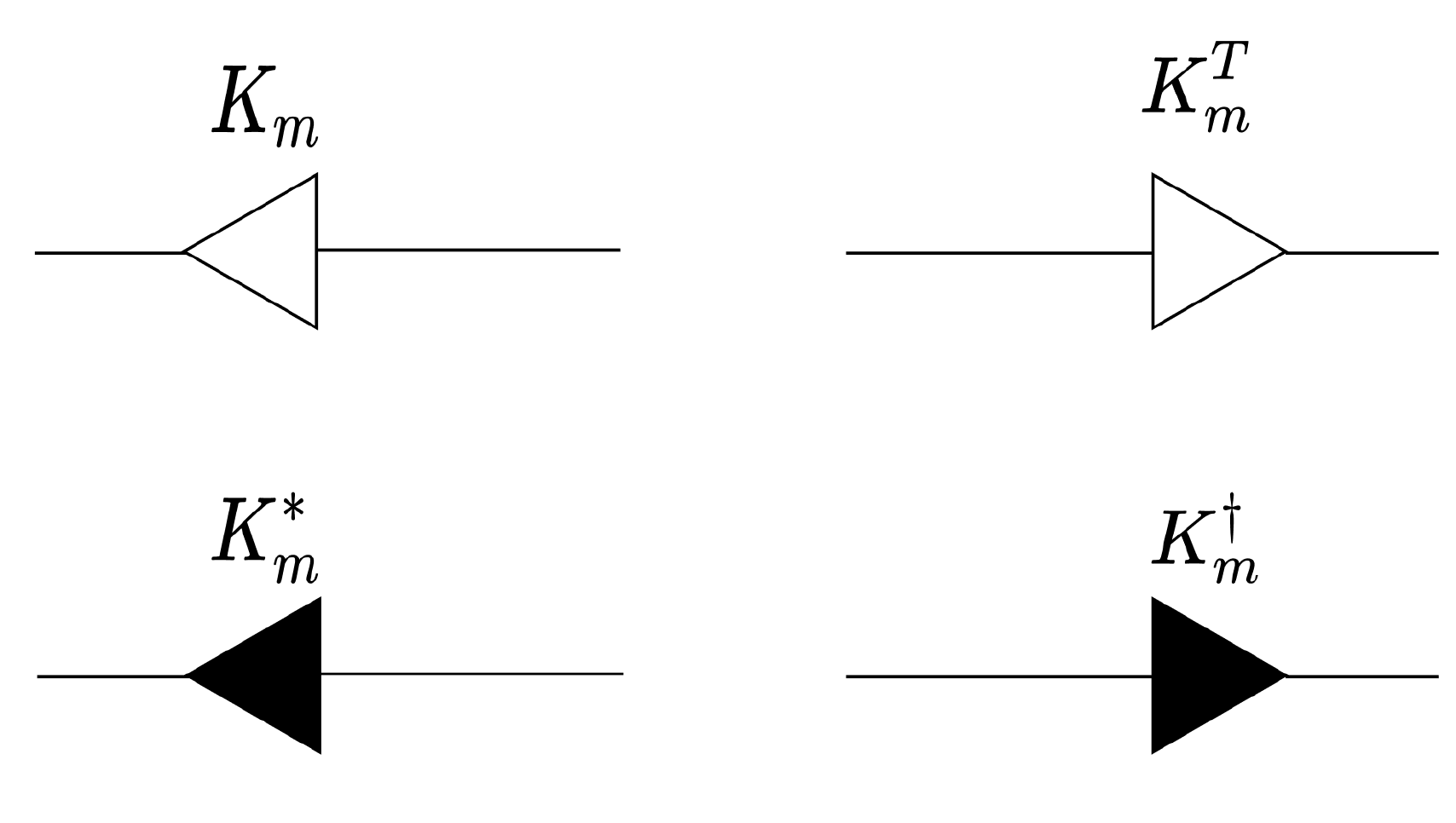}}}
\end{align}
The transposition of the Kraus operator is represented with the right pointing triangle, and the complex conjugate is represented with the black color. We always omit summations such as $\sum_{m,i,j}$ in the diagrammatic equation. When the channels are applied multiple times, we assign a number to the Kraus operators to distinguish $K_m$ and $K_{m'}$. When there is no confusion, we omit the numbers assigned to the Kraus operators. 

Using the above notations, it is now easy to see that the Liouville representation and the Choi state of the channel are related to each other via the reshuffling operation. 
\begin{align}
\begin{dcases}
    \mathcal{E}_B = \frac{1}{d}\sum_{m, i,j} K_m|i\rangle \langle j|K_m^\dagger \otimes |i\rangle \langle j| = \frac{1}{d} \, \vcenter{\hbox{\includegraphics[clip,scale=0.2]{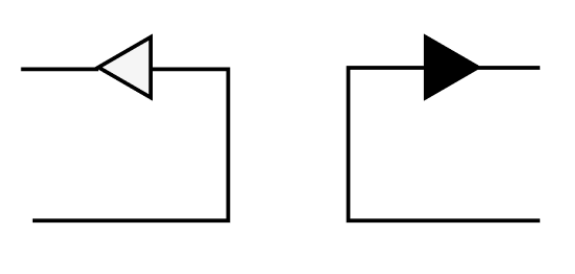}}} \\
    \mathcal{E}_A = \sum_{m} K_m \otimes K_m^* = d\mathcal{R}(\mathcal{E}_B) = \quad \vcenter{\hbox{\includegraphics[clip,scale=0.15]{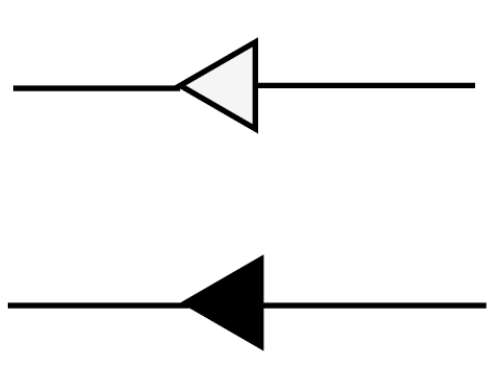}}}
\end{dcases}
\end{align}
\noindent Our key observation in the main text Eq.(11) can be easily seen from the diagram notation;
\begin{align}
    \rho_\pm = \frac{1}{d} \mqty[\vcenter{\hbox{\includegraphics[clip,scale=0.2]{figD.pdf}}}  & \pm \vcenter{\hbox{\includegraphics[clip,scale=0.2]{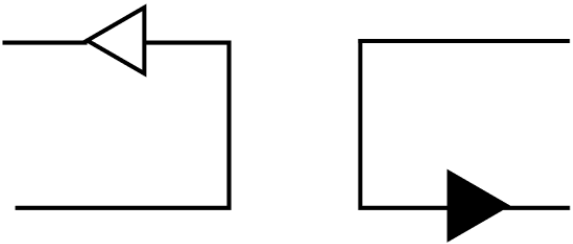}}} \\ \pm \vcenter{\hbox{\includegraphics[clip,scale=0.2]{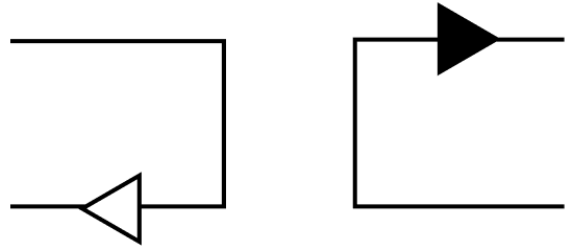}}} & \vcenter{\hbox{\includegraphics[clip,scale=0.2]{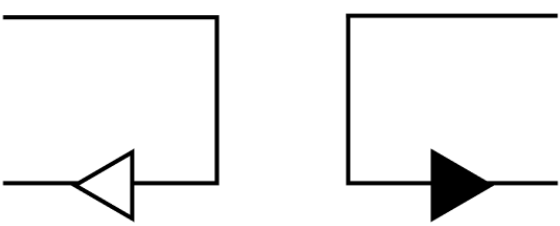}}} ] = \frac{1}{d} \mqty[\vcenter{\hbox{\includegraphics[clip,scale=0.2]{figD.pdf}}}  & \pm \qty(\vcenter{\hbox{\includegraphics[clip,scale=0.2]{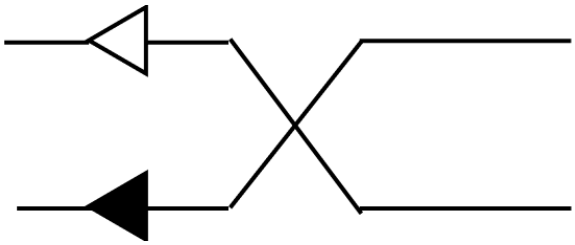}}})^{T_Z} \\ \pm \qty(\vcenter{\hbox{\includegraphics[clip,scale=0.2]{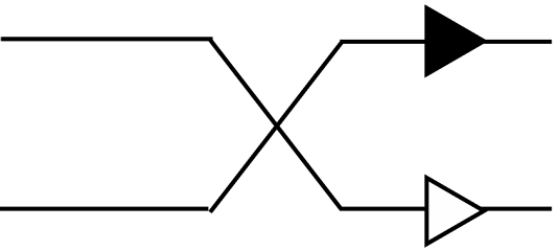}}})^{T_Z} & \vcenter{\hbox{\includegraphics[clip,scale=0.2]{figL.pdf}}} ]. 
\end{align}

\section{Basic properties of a quantum channel}\label{app:channelprop}
Here we briefly review some basic properties of a quantum channel. As mentioned in the main text, the eigenvalue of a quantum channel $\mathcal{E}: \mathfrak{L}(\mathcal{H}) \rightarrow \mathfrak{L}(\mathcal{H})$ is defined as 
\begin{align}
    \mathcal{E}(\rho_n^R) &= \lambda_n \rho_n^R, \quad (\rho_n^R \in \mathfrak{L}(\mathcal{H}))
\end{align}
where $\lambda_n \in \mathbb{C}$. This also implies 
\begin{align}
    \mathcal{E}^\dagger(\rho_n^L) &= \lambda_n^* \rho_n^L, \quad (\rho_n^L \in \mathfrak{L}(\mathcal{H})), 
\end{align}
where $\rho_n^R, \rho_n^L$ are sometimes referred to as the right-eigenvector and left-eigenvector, respectively. They satisfy 
\begin{equation}\label{eq:biorthogonal}
\Tr [(\rho_m^L)^\dagger\rho_n^R]=0 \quad (m\neq n), 
\end{equation}
which allows one to decompose the action of the channel as
\begin{equation}
    \mathcal{E}(\rho) = \sum_n c_n \lambda_n \rho_n^R.
\end{equation}
This is essentially a spectrum problem of a \textit{non-Hermitian matrix}, and some notable differences arise from the Hermitian case. For example, in the Hermitian matrix case, eigenvectors corresponding to different eigenvalues are always orthogonal to each other, but this is not the case for non-Hermitian matrices. Below are some basic properties of the eigenvalues of quantum channels. 

\begin{lm}\label{lm:channeleigval}
The eigenvalue and eigenoperator of quantum channels have the following properties; 
\begin{enumerate}
    \item $\mathcal{E}$ has eigenvalue 1. 
    \item If $\lambda_n$ is an eigenvalue, $\lambda_n^*$ is also an eigenvalue. Also, right eigenoperator is traceless except for the fixed point. 
    \item $|\lambda_n|\leq 1$.
\end{enumerate}
\end{lm}

\begin{proof}
1; The following proof is given in~\cite{Watrous_2018}. Define a CPTP map $\Phi$ as 
\begin{equation}
    \Phi_n(\rho) = \frac{1}{2^n}\sum_{k=0}^{2^n-1} \mathcal{E}^{(k)}(\rho).
\end{equation}
We define the following set $D_n \equiv \{ \Phi_n(\rho) | \rho \in \mathcal{L}(\mathcal{H}) \}$. It holds that 
\begin{equation*}
    \Phi_{n+1}(\rho) = \frac{1}{2^{n+1}}\sum_{k=0}^{2^{n+1}-1} \mathcal{E}^{(k)}(\rho) = \frac{1}{2}\Phi_{n}(\rho) + \frac{1}{2}\Phi_{n}(\mathcal{E}^{(n)}(\rho)) = \Phi_{n}\qty(\frac{1}{2}\rho + \frac{1}{2}\mathcal{E}^{(n)}(\rho))
\end{equation*}
and thus $D_{n+1} \subseteq D_n$, so there exist $\rho_0 \in D_0 \bigcap D_1 \bigcap \cdots$. Then, for arbitrary positive integer $n$, $\rho_0=\Phi_n(\sigma)$ holds for some $\sigma\in \mathcal{L}(\mathcal{H})$. It follows that 
\begin{equation*}
    \mathcal{E}(\rho_0) - \rho_0 = \mathcal{E}(\Phi_n(\sigma))-\Phi_n(\sigma) = \frac{\mathcal{E}^{(2^n)}(\sigma)-\sigma}{2^n}.
\end{equation*}
Thus, 
\begin{equation}
    \lVert \mathcal{E}(\rho_0) - \rho_0 \rVert_1 \leq \frac{1}{2^{n-1}}
\end{equation}
for arbitrary $n$, meaning that $\mathcal{E}(\rho_0)=\rho_0$. 
\end{proof}
\begin{proof}
2;
$\mathcal{E}$ is Hermiticity preserving, so 
\begin{equation}
    \mathcal{E}((\rho_n^R)^\dagger) = (\mathcal{E}(\rho_n^R))^\dagger = \lambda_n^*(\rho_n^R)^\dagger
\end{equation}
The right eigenoperator is traceless $\Tr \rho_n^R=0$ except for the fixed point, since $\Tr \mathcal{E}(\rho) = \Tr \rho$. 
\end{proof}
\begin{proof}
3; 
Trace norm is monotonically decreasing under trace-preserving positive maps (see \cite{Watrous_2018}, 3.40)
\begin{equation}
    |\lambda_n| = \frac{\lVert \lambda_n \rho_n^R \rVert_1}{\lVert \rho_n^R \rVert_1} = \frac{\lVert \mathcal{E}(\rho_n^R)\rVert_1}{\lVert \rho_n^R\rVert_1}  \leq 1. 
\end{equation}
\end{proof}

\noindent Next, we prove some lemmas on the norm of quantum channels. 

\begin{lm}\label{lm:2ndmomentpurity}
The second singular value moment of a quantum channel is equivalent to the purity of its Choi state;
\begin{align}
    S_2 = \Tr (\mathcal{E}_B)^2
\end{align}
\end{lm}

\begin{proof}
\begin{align}
    S_2 := \frac{1}{d^2} \sum_i \sigma_i^2 &= \Tr \qty[ \mathcal{R}\qty(\mathcal{E}_B)^\dagger \mathcal{R}\qty(\mathcal{E}_B)]\nonumber\\
    &= \frac{1}{d^2} \, \vcenter{\hbox{\includegraphics[clip,scale=0.3]{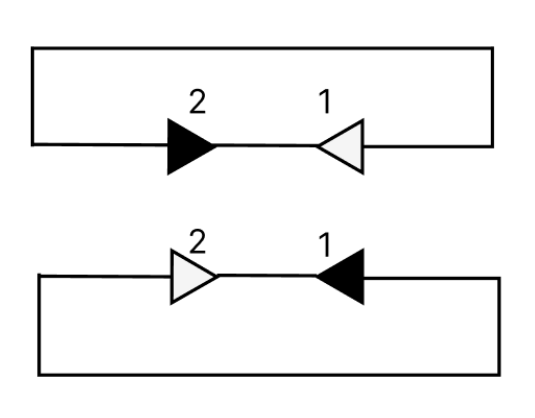}}} \nonumber\\
    &= \frac{1}{d^2} \, \vcenter{\hbox{\includegraphics[clip,scale=0.3]{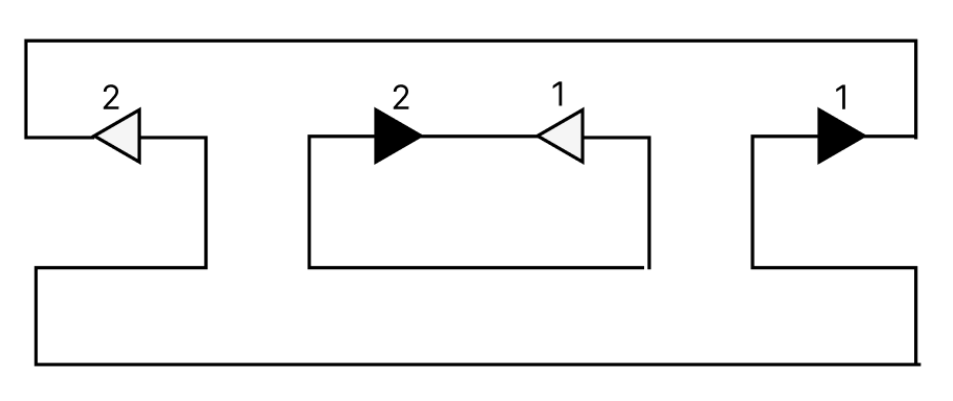}}}\nonumber \\
    &= \Tr \qty(\mathcal{E}_B)^2. 
\end{align}
\end{proof}

\begin{lm}\label{lm:channelopnorm}
Consider a quantum channel $\mathcal{E}$ acting on a $d$-dimensional system. The operator norm of its Liouville representation satisfies~\cite{Wallman_2014} 
    \begin{equation}
        1 \leq  \lVert \mathcal{E}_A \rVert_\infty \leq \sqrt{d}.
    \end{equation}
The lower bound is saturated by unital channels, while the upper bound is saturated by trace and replace channels such as $\mathcal{E}(\rho) = \Tr(\rho) |0\rangle \langle 0|$.  
\end{lm}

\begin{proof}
As mentioned in Lemma~\ref{lm:channeleigval}, every quantum channel has a fixed point. The spectral radius $r = \max \{ |\lambda_i|\}$ for any square matrix $A$ gives a lower bound on the spectral norm. For quantum channels, the spectral radius is 1, which means
\begin{equation}
    1 \leq \lVert \mathcal{E}_A\rVert_{\infty} = \max_i \sigma_i.
\end{equation}
The equality holds if and only if $\mathcal{E}$ is unital (see Theorem 4.27 of~\cite{Watrous_2018}) Next, consider the Frobenius norm 
\begin{equation}
    \lVert A \rVert_F = \sqrt{\sum_{ij}|a_{ij}|^2}.
\end{equation}
For a vectorized density matrix 
\begin{equation}
    |\rho\rangle \rangle = \sum_{ij} \rho_{ij}\ket{i} \otimes \ket{j},
\end{equation}
This corresponds to the usual inner product of vectors. Thus, it can be shown that
\begin{equation}
    \lVert \mathcal{E}_A \rVert_{\infty} = \sup_{\rho \neq 0} \frac{\lVert |\mathcal{E}(\rho)\rangle \rangle \rVert_2}{\lVert 
    |\rho\rangle \rangle \rVert_2} = \sup_{\rho \neq 0} \frac{\lVert \mathcal{E}(\rho)\rVert_F}{\lVert \rho \rVert_F} = \sup_{\rho \neq 0} \sqrt{\frac{\Tr \mathcal{E}(\rho)^2 }{\Tr \rho^2}} \leq \sqrt{d},
\end{equation}
which follows from the fact that $\frac{1}{d} \leq \Tr \rho^2 \leq 1$. Therefore, we conclude that 
\begin{equation}
    1 \leq \max_i \sigma_i = \lVert \mathcal{E}_A \rVert_\infty \leq \sqrt{d}.
\end{equation}
The upper bound is saturated if $\rho=I/d$, and $\mathcal{E}(\rho)$ is a pure state. This is achieved, for instance, by the trace and replace channel $\mathcal{E}(\rho)=\Tr(\rho) |0\rangle \langle 0|$. 
\end{proof}
\begin{lm}\label{lm:channeltrnorm}
Consider a quantum channel $\mathcal{E}$ acting on a $d$-dimensional system. The trace norm of its Liouville representation satisfies
    \begin{align}
        \lVert \mathcal{E}_A\rVert_1 \leq d^2. 
    \end{align}
\end{lm}
\begin{proof}
    From Lemma~\ref{lm:2ndmomentpurity}, we have 
    \begin{align*}
        \lVert \mathcal{E}_A \rVert_2 = \sqrt{\sum_{i=1}^{d^2}\sigma_i^2} = \sqrt{d^2 \Tr \qty(\mathcal{E}_B)^2} \leq d
    \end{align*}
    From the Cauchy-Schwartz inequality $||\mathcal{E}_A||_1 \leq d ||\mathcal{E}_A||_2$, we obtain 
    \begin{align}
        \lVert \mathcal{E}_A \rVert_1 \leq d^2. 
    \end{align}
    The upper bound is saturated by unitary channels. \\
\end{proof}

\noindent Finally, we also bound the operator norm of arbitrary reshuffled states. 
\begin{lm}\label{lm:stateopnorm}
Let $\mathcal{R}$ denote the reshuffling operation. For a bipartite density matrix $\rho$ in a $D=d^2$ dimensional system, 
    \begin{align}
        \lVert \mathcal{R}(\rho) \rVert_\infty \leq 1. 
    \end{align}
\end{lm}
\begin{proof}
    \begin{align}
        \lVert \mathcal{R}(\rho)\rVert_\infty \leq \lVert \mathcal{R}(\rho)\rVert_2 = \sqrt{\Tr \rho^2} \leq 1. 
    \end{align}
The upper bound $\lVert \mathcal{R}(\rho)\rVert_\infty=1$ is saturated by pure product states. 
\end{proof}

\section{Review of some quantum algorithm tools}\label{app::algorithmtools} 
Here we review some algorithmic tools used in the main text.

\subsection{State exponentiaion and partially transposed state exponentiation}
State exponentiation~\cite{Lloyd_2014} is a protocol that approximately applies a unitary operator $e^{-i\rho t}$, given multiple copies of $\rho$. Let $\rho \in \mathfrak{S}(\mathcal{H}_X)$ and $\sigma \in \mathfrak{S}(\mathcal{H}_{X'})$. The key idea is to utilize the following relation; 
\begin{align}
    \Tr_{X} [e^{-i\mathbb{F}_{XX'} \Delta t} (\rho \otimes \sigma  )e^{i\mathbb{F}_{XX'} \Delta t}] &= (\cos ^2 \Delta t) \sigma + (\sin ^2 \Delta t) \rho -i\sin \Delta t \cos \Delta t [\rho, \sigma]\nonumber\\
    &= \sigma -i\Delta t [\rho, \sigma] + O(\Delta t^2)\nonumber\\
    &= e^{-i\rho \Delta t}\sigma e^{i\rho \Delta t}+O(\Delta t^2).
\end{align}
$\mathbb{F}_{XX'}$ denotes the \texttt{SWAP} operator acting on the Hilbert spaces $\mathcal{H}_X$ and  $\mathcal{H}_{X'}$. In the first line, we have used the relation
\begin{align}
    \Tr_{X} [(\rho\otimes \sigma) \mathbb{F}_{XX'}] = \sigma \rho, \quad \Tr_{X} [\mathbb{F}_{XX'}(\rho \otimes \sigma)] = \rho \sigma. 
\end{align}
Choosing $\Delta t = \delta$ and repeating this procedure for $O(\frac{1}{\delta})$ times, we obtain a $\delta$-close approximation of the unitary $e^{-i\rho t}$. 

In a similar spirit,  given many copies of a bipartite state $\rho \in \mathfrak{S}(\mathcal{H}_A \otimes \mathcal{H}_B)$, \textit{partially transposed state exponentiation}~\cite{Wei:2024znk} is a protocol that appoximately applies the unitary $e^{-i\rho^{T_B}t}$ on a state $\sigma \in \mathfrak{S}(\mathcal{H}_{A'} \otimes \mathcal{H}_{B'})$. This is achieved by applying $e^{-i\widetilde{H} \Delta t}$ with $\widetilde{H}:=\mathbb{F}_{AA'} \otimes d|\Phi_d^+\rangle \langle \Phi_d^+|_{BB'}$ to $\rho \otimes \sigma$ and tracing out the Hilbert space of $\rho$, which gives
\begin{align}
  \Tr_{AB}[e^{-i\widetilde{H}\Delta t} (\rho\otimes \sigma) e^{i\widetilde{H}\Delta t}]
  &= \sigma -i\Delta t[\rho^{T_B},  \sigma]+O(d\Delta t^2)\nonumber,\\
  &=  e^{-i\rho^{T_B} \Delta t} \sigma e^{i\rho^{T_B} \Delta t} +O(d\Delta t^2).  
\end{align}
Here, the \texttt{SWAP} operator $\mathbb{F}_{AA'}$ acts on $\mathcal{H}_A \otimes \mathcal{H}_{A'}$ and the unnormalized MES $d|\Phi_d^+\rangle \langle \Phi_d^+|_{BB'}$ acts on $\mathcal{H}_{B} \otimes \mathcal{H}_{B'}$. The first line uses the relations 
\begin{align}
    \Tr_{AB} [(\rho\otimes \sigma) \tilde{H}] = \sigma \rho^{T_B}, \quad \Tr_{AB} [\tilde{H}(\rho \otimes \sigma)] = \rho^{T_B} \sigma. 
\end{align}
Using a diagramatic notation, the first relation for example can be understood as follows
\begin{figure}[H]
\centering
\includegraphics[width=0.6\columnwidth]{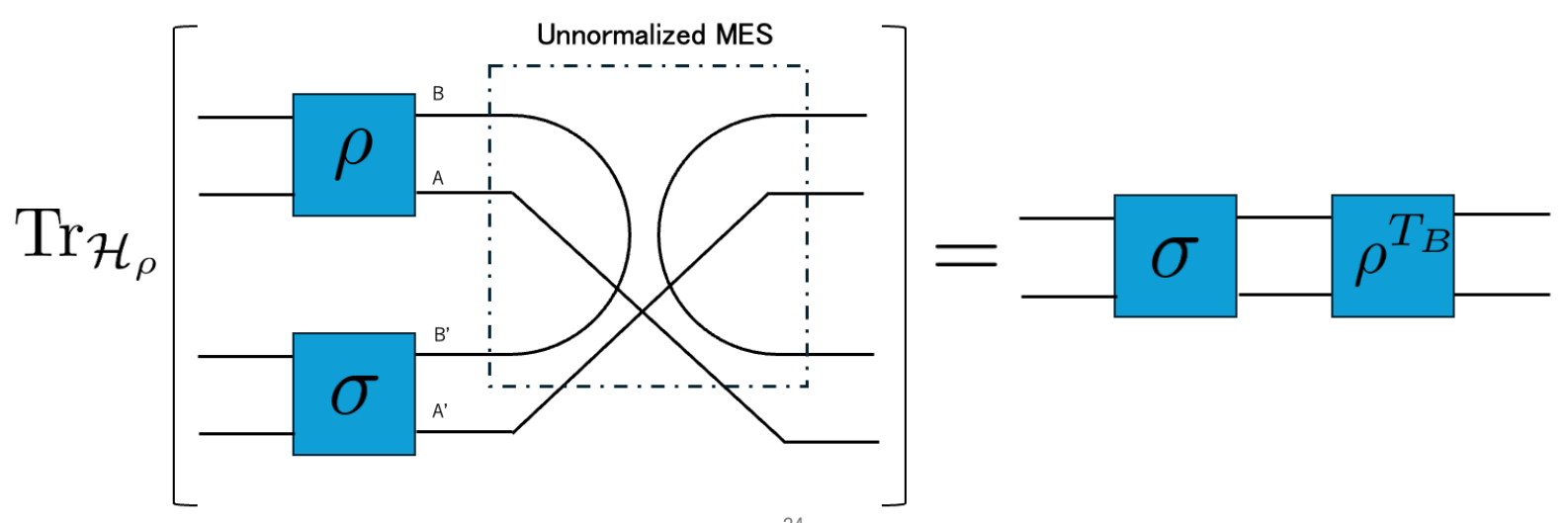}
\label{fig:partialtranspose}
\end{figure}

\subsection{Quantum singular value transformation (QSVT)}
The quantum singular value transformation (QSVT) is a framework that encompasses many well-known algorithms such as amplitude amplification, matrix inversion, and Hamiltonian simulation \cite{Gily_n_2019, mrtc_unification_21}. Specifically, one can show that the core results of quantum signal processing (QSP) \cite{lyc_16_equiangular_gates, lc_17_simulation, Low:2016znh}, a statement about products of parameterized SU(2) elements, can be `lifted' to describe permissible manipulation of linear operators encoded into sub-blocks of a larger unitary. While the details of this argument are somewhat precise, relying on properties of the cosine-sine decomposition exposited quite cleanly in \cite{cs_qsvt_tang_tian}, we quote important definitions and results here for the unfamiliar reader. To start, we reproduce the commonly cited definition of a block-encoding given in the main text.

\begin{dfn}($(\alpha, a, \epsilon)$ block encoding).
    If $A$ is an $n$-qubit operator and $U$ is an $(n+a)$ qubit unitary operator that satisfies
    \begin{align}
        \lVert A-\alpha(\bra{0}^{\otimes a}\otimes I) U (\ket{0}^{\otimes a}\otimes I) \rVert \leq \epsilon,
    \end{align}
    then we say that $U$ is an $(\alpha, a, \epsilon)$\textit{-block encoding} of $A$. In other words, the $(2\times2)$-block unitary $U$ contains an $\alpha$-subnormalized copy of an $\varepsilon$-approximation to $A$ in its top-left block (indexed by $|0\rangle^{\otimes a}$ for both row and column).
\end{dfn}

Given the definition above, we can define precisely what we mean by transforming the singular values of $A$, after which we can provide a useful workhorse theorem from QSVT, which will be all that is required in our application. Specifically, given the simplest case of a $(1, 1, 0)$ block-encoding of some linear operator $A$, QSVT performs the transformation
    \begin{equation}
        \begin{bmatrix}
            A & \ast \\
            \ast & \ast 
        \end{bmatrix}
        \longmapsto
        \begin{bmatrix}
            P^{(SV)}(A) & \ast \\
            \ast & \ast 
        \end{bmatrix},
    \end{equation}
where the blocks here have been implicitly labeled by a single auxiliary qubit state, or more specifically
    \begin{equation}
        A = \sum_{\xi} \xi_k |\ell_k\rangle\langle r_k|
        \mapsto 
        P^{(SV)}(A) :=
        \sum_{\xi} P(\xi_k) |\ell_k\rangle\langle r_k|,
    \end{equation}
with $\xi_k$, $|\ell_k\rangle$, $|r_k\rangle$ the $k$-th singular value, left singular vector, and right singular vector of $A$, respectively. The key insight and surprise of QSVT will be that we can induce this transformation on the singular values without having to know the singular values \emph{or} singular vectors \footnote{In truth we have to know milder things, like a bound on the range of possible singular values, and the span of the left and right singular vectors respectively, but these are often trivially known.}, meaning these manipulations can, in many common cases, be extremely efficient. Moreover, while the referenced theorem below makes mention only of query complexity to the block encoding, the additional required gates are simple and their number scales linearly with the query complexity.

The following theorem, which describes permissable transformations of the singular values of block encoded (sub-normalized, approximate, Hermitian) linear operators, is necessary and sufficient for our applications. Note that extensive results exist on the construction and manipulation of block encodings, some of which may be useful in extensions to this work; the interested reader is directed to common pedagogical texts \cite{Gily_n_2019, mrtc_unification_21, cs_qsvt_tang_tian}.

\begin{thm}(QSVT for Hermitian matrices~\cite{Gily_n_2019}) Suppose $U$ is an $(\alpha, a, \epsilon)$-block encoding of a Hermitian matrix $A$, and $P(x) \in \mathbb{R}[x]$ is a polynomial that satisfies $|P(x)| \leq \frac{1}{2}$ for all $x \in [-1, 1]$. Then, we can implement a unitary $\widetilde{U}$ that is a 
$(1, a+2, 4(\deg P)\sqrt{\frac{\epsilon}{\alpha}}+\delta)$-block-encoding of $P(A/\alpha)$, using $O(\deg P)$ queries to $U$. The description of the circuit can be computed with a classial computer in time $O(\mathrm{poly} (\deg P, \log \frac{1}{\delta}))$. Note also that if the parity of $P(x)$ is definite, then the same statement holds for $|P(x)| \leq 1$ on $x \in [-1, 1]$.
\end{thm}

To calculate spectral moments, we need a suitable polynomial approximation of the power function $x^q$. The most common choice for the approximating polynomial is the degree $d$ truncated Chebyshev polynomial
\begin{align}
    \tilde{P}_d(x) \equiv \frac{c_0}{2} + \sum_{k=1}^d c_k T_k(x), 
\end{align}
which can exhibit poor approximation around $x = 0$. For our setting, a more well-behaved approximation (i.e., one which uniformly converges across all $x \in [-1,1]$ at a near optimal rate) will be obtained by averaging over $d$ such polynomials 
\begin{align}
    P_{d'}(x) = \frac{1}{d} \sum_{k=d}^{2d-1} \tilde{P}_d(x).
\end{align}
This modified polynomial was employed in~\cite{Liu_2025} to obtain an estimate for the Tsallis entropy of states to avoid having to make assumptions on the spectrum of the density matrix (e.g., bounds on its condition number). In particular, for a fixed positive integer $r$ and a real number in $(-1,1)$, there is a degree $d=\lceil(\beta/\epsilon)^\frac{1}{r+\alpha} \rceil$ polynomial for any $\epsilon \in (0, \frac{1}{2}]$ such that 
\begin{align}
    \max_{x \in [-1,1]} \left\lvert\frac{1}{2}x^{r-1}|x|^{1+\alpha}-P_d(x)\right\rvert \leq \epsilon, \quad \max_{x \in [-1,1]} \lVert P_d(x)\rVert \leq 1. 
\end{align}
The same reference also notes that a polynomial of degree $O(1/\epsilon''^{\frac{1}{q-1}})$ is generally required to obtain the desired $\epsilon''$-\emph{uniform} approximation to $\frac{1}{2}x^{q-1}$ over $x \in [-1,1]$.

\subsection{Samplizer}
Here we review the notion of the samplizer, which generalizes the density matrix exponentiation trick~\cite{Lloyd_2014}, originally introduced in~\cite{wang_et_al:LIPIcs.ESA.2024.101}. This method converts a block-encoding unitary access model into a quantum state access model.

\begin{thm}(Samplizer) Suppose $C = \{C[U]\}$ is a quantum circuit family with access to $(1,m,0)$-block encoding of $\rho/2$ with query complexity Q. If $m\geq 4$, then, for every $\delta>0$, there is a quantum channel family
$\mathrm{Samplize}_\delta\langle C \rangle$ with sample access to $\rho$ with sample complexity $O\qty(\frac{Q^2}{\delta}\log^2\qty(\frac{Q}{\delta}))$ satisfying for every $\rho$, there is a specific unitary operator $U_\rho$ that is a $(2, m, 0)$-block-encoding of $\rho$ such that
\begin{align}
    \lVert\mathrm{Samplize}_\delta\langle C \rangle(\rho) - C[U_\rho]\rVert_\diamond \leq \delta
\end{align}
\end{thm}

\begin{proof}
Suppose the quantum circuit has the following form
\begin{align}
    C[U] = G_Q U_Q \cdots G_2 U_2 G_1 U_1 G_0
\end{align}
where $G_i$ consists of one and two qubit gates, and $U_i$s are either unitary, controlled unitary or their inverse. Utilizing the state exponentiation protocol in \cite{Lloyd_2014}, one can show that for $\epsilon=\delta/Q$ we can implement the channel $\mathcal{E}, \mathcal{E}^{\mathrm{inv}}$ up to precision $\epsilon$
\begin{align}
    \lVert\mathcal{E}_\rho - \mathcal{U}_\rho\rVert_\diamond \leq \epsilon, \quad\lVert\mathcal{E}_\rho^{\mathrm{inv}} - \mathcal{U}^{\mathrm{inv}}_\rho\rVert_\diamond \leq \epsilon
\end{align}
where $\mathcal{U}(\sigma) = U_\rho \sigma U_\rho^\dagger, \, \mathcal{U}^{\mathrm{inv}}(\sigma) = U_\rho^\dagger \sigma U_\rho$, using $O(\frac{1}{\epsilon} \log^2 \qty(\frac{1}{\epsilon}))$ samples. By appropriately enlarging the Hilbert space, one can show that 
\begin{align}
    \lVert C'[\rho]-C[U_\rho\otimes I^{\otimes (m-4)}]\rVert_\diamond \leq Q\epsilon = \delta
\end{align}
since 
\begin{align}
    \lVert\mathcal{E}_1 + \mathcal{E}_2\rVert_\diamond \leq \lVert\mathcal{E}_1\rVert_\diamond + \lVert\mathcal{E}_2\rVert_\diamond, \quad 
    \lVert\mathcal{E}_2\mathcal{E}_1\rVert_\diamond \leq \lVert\mathcal{E}_2\rVert_\diamond \lVert\mathcal{E}_1\rVert_\diamond
\end{align}
\end{proof} 

\subsection{The Hadamard test}
\noindent The Hadamard test calculates $\Tr (\rho U)$ by repeatedly running the following circuit.
\begin{figure}[H]
\centering
\includegraphics[width=0.35\columnwidth]{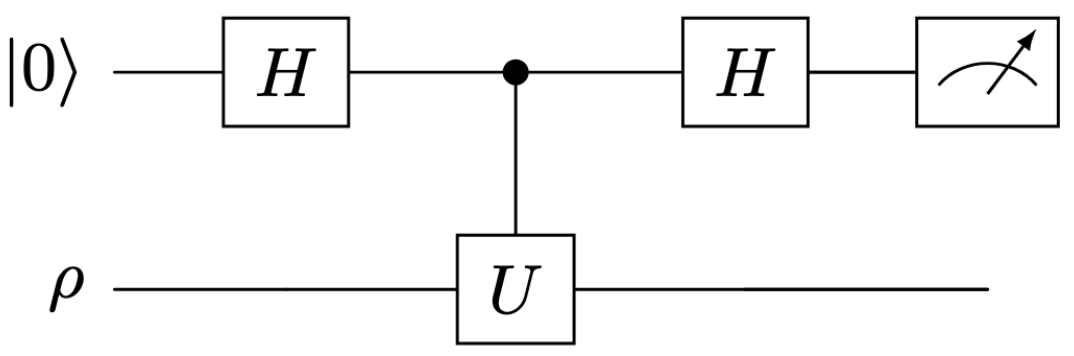}
\label{fig:Hadamard}
\end{figure}
\noindent More specifically, the binary measurement outcome obeys the probability distribution
\begin{align}
    p(0) = \frac{1+\mathrm{Re}[\Tr (\rho U)]}{2}\nonumber\\
    p(1) = \frac{1-\mathrm{Re}[\Tr (\rho U)]}{2}
\end{align}
so if we record $+1$ when the measurement outcome is $0$, and $-1$ when the measurement outcome is $-1$, taking the average produces an estimate $\hat{X}$ for $\mathrm{Re} [\Tr (\rho U)]$; $\mathbb{E}[\hat{X}] = \Tr (\rho U)$. When the input $\rho$ is the maximally mixed state, it is also called the DQC1~\cite{PhysRevLett.81.5672} circuit. Repeating the circuit $N$ times, we know from the Chernoff bound that
\begin{align}
    \mathrm{Pr}(|\hat{X}-\mathbb{E}[\hat{X}]|\leq \epsilon) \leq 1-2\mathrm{exp}(-2N\epsilon^2). 
\end{align}
Thus, $O(\frac{\log \qty(\frac{1}{\delta})}{\epsilon^2})$ repetition of the circuit allows one to estimate $\Tr (\rho U)$ up to $\epsilon$-precision with probability $P \geq 1-\delta$. Moreover, if $U$ is an $n+a$ qubit unitary that block-encodes an $n$-qubit operator $A$, $\Tr (\rho A)$ can be analogously estimated up to a fixed precision with $O(\frac{\log \qty(\frac{1}{\delta})}{\epsilon^2})$ repetitions. 

\section{Block-encoding via reshuffling circuit}\label{ap:stateBE}
Here we discuss a different scheme for approximately block-encoding the Hermitized Liouville representation. Our overall scheme proceeds by first constructing a \texttt{SWAP} circuit implementing the \textit{reshuffling} operation in a Hermitized form, followed by density matrix exponentiation. We then see how the reshuffling circuit is useful for directly estimating the full-singular value spectrum of unital quantum channels. 
\begin{thm}\label{thm6} 
    Given a $D=d\times d$-dimensional bipartite state $\rho\in\mathbb{C}^{D\times D}$, we can construct a quantum channel $\widetilde{\mathcal{U}}$ that is $\delta$-close in diamond norm to a $(1, 3, 0)$-block encoding unitary of $\frac{2}{\pi}H$, where 
    \begin{align}
        H:= \frac{1}{2d^{1-k}}\mqty[ O & \mathcal{R}(\rho)\\ \mathcal{R}(\rho)^\dagger & O]
    \end{align}
    using $\widetilde{O}(\frac{d^{2k}}{\delta} \log^2 \frac{1}{\delta})$ copies of $\rho$. In particular, when $\rho=\mathcal{E}_B$, namely, when $\rho$ is the Choi state of the unknown channel $\mathcal{E}$, $H$ is the Hermitized Liouville representation. $k\in [0,1]$, $k \in [0, \frac{3}{2}]$, and $k \in [0, 2]$ ensures $\lVert H \rVert_\infty \leq \frac{1}{2}$ for arbitrary states, Choi states of general channels, and Choi states of unital channels, respectively.
\end{thm}

\begin{proof}
A key observation is that the reshuffling operation $\mathcal{R}$ can be obtained via the relation
\begin{align}
    \frac{1}{d}\mathcal{R}(\rho)  = \Tr_{23}\qty[\qty(\rho \otimes |\Phi_d^+\rangle \langle \Phi_d^+|) \mathbb{F}_{23} \mathbb{F}_{14} \mathbb{F}_{34}]. 
\end{align}
Here, $\mathbb{F}_{ij}$ represents the \texttt{SWAP} operator between system $i$ and $j$. The corresponding circuit is 
\begin{figure}[H]
\centering
\includegraphics[width=0.35\columnwidth]{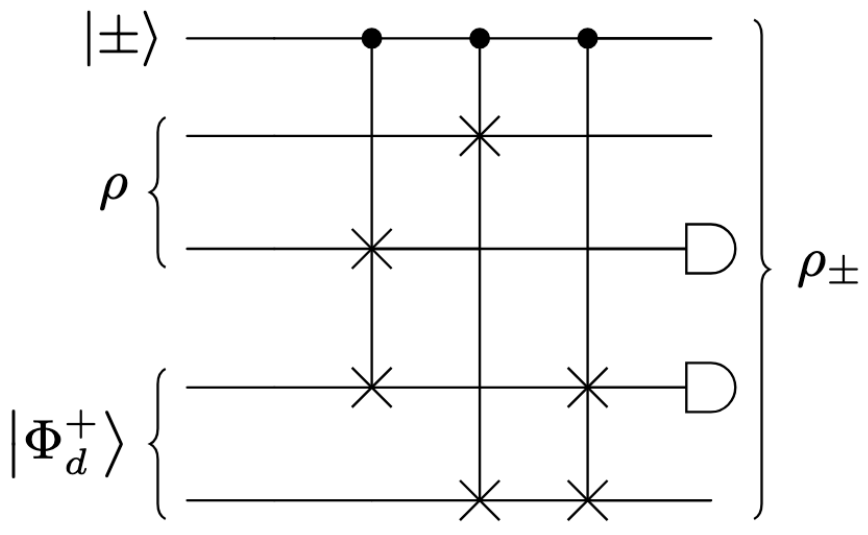}
\label{fig:reshuffling}
\end{figure}
\noindent where the white caps represent the trace and the subsystem labels read from top to bottom, following the control qubit $C$ which labels the uppermost system. This circuit prepares a state $\tilde{\rho}_\pm \in \mathfrak{S}(\mathbb{C}^2 \otimes \mathbb{C}^d \otimes \mathbb{C}^d)$ given by
\begin{align}\label{eq:reshuffling}
    \tilde{\rho}_\pm &= |0\rangle \langle 0| \otimes \frac{1}{2d}\Tr_2 [\rho] \otimes I \pm |0\rangle \langle 1| \otimes \frac{1}{2d} \mathcal{R}(\rho) \nonumber\\
    &\pm |1\rangle \langle 0| \otimes \frac{1}{2d} \mathcal{R}(\rho)^\dagger +  |1\rangle \langle 1| \otimes \frac{1}{2d} I \otimes \Tr_1 [\rho]. 
\end{align}
\noindent Next, we use \textit{density matrix exponentiation}~\cite{Lloyd_2014} to approximate the unitary $e^{-i\rho_\pm t}$ from many copies of $\rho_{\pm}$ as
\begin{align}
    \Tr_1 [e^{-i\mathbb{F} \Delta t} (\tilde{\rho}_\pm \otimes \sigma  )e^{i\mathbb{F} \Delta t}] &= (\cos ^2 \Delta t) \sigma + (\sin ^2 \Delta t) \tilde{\rho}_\pm -i\sin \Delta t \cos \Delta t [\tilde{\rho}_\pm, \sigma]\nonumber\\
    &= \sigma -i\Delta t [\tilde{\rho}_\pm, \sigma] + O(\Delta t^2)\nonumber\\
    &= e^{-i\tilde{\rho}_\pm \Delta t}\sigma e^{i\tilde{\rho}_\pm \Delta t}+O(\Delta t^2).
\end{align}
This allows us to obtain an approximation of the unitary
\begin{align}
    \widetilde{U}_{\mathcal{R}(\rho)} &:= e^{i \tilde{\rho}_- \Delta t}e^{-i\tilde{\rho}_+ \Delta t}\nonumber\\
    &= \exp \qty(-\frac{i}{d}\mqty[O & \mathcal{R}(\rho) \\ \mathcal{R}(\rho)^\dagger & O]\Delta t)+O(\Delta t^2).
\end{align}
Choosing $\Delta t = \frac{2\delta}{d^{k}}$ and applying the above procedure $O(\frac{d^{2k}}{4\delta})$ times, we arrive at a $O(\Delta t^2) \times O(\frac{d^{2k}}{4\delta}) = \delta$-close approximation of the unitary 
\begin{align}
    U_{\mathcal{R}(\rho)} = \exp \qty(-\frac{i}{2d^{1-k}}\mqty[O & \mathcal{R}(\rho) \\ \mathcal{R}(\rho)^\dagger & O]) = e^{-iH}.
\end{align}
In particular, when $\rho$ is the Choi state $\mathcal{E}_B$, $\mathcal{R}(\mathcal{E}_B)=\frac{1}{d}\mathcal{E}_A$, so $H$ is the \emph{Hermitized Liouville representation} of the unknown channel $\mathcal{E}$. Note that control-$U_{\mathcal{R}(\rho)}$ can be similarly obtained by replacing $\tilde{\rho}_\pm$ with $|1\rangle \langle 1| \otimes \tilde{\rho}_{\pm}$~\cite{Kimmel_2017}. Now, $H$ satisfies $\lVert H\rVert_\infty \leq \frac{1}{2}$ as long as $k\leq 1$ (see Appendix \ref{app:channelprop}, Lemma~\ref{lm:stateopnorm}). When restricted to Choi states, $k\leq \frac{3}{2}$ and $k\leq 2$ ensures $\lVert H\rVert_\infty \leq \frac{1}{2}$ for general channels and unital channels, respectively (see Appendix \ref{app:channelprop}, Lemma~\ref{lm:channelopnorm}). Then, following the exact same argument provided in the proof of Theorem 1 in the main text, we can obtain a quantum channel $\widetilde{\mathcal{U}}$ that is $\delta$-close in diamond norm to a $(1, 3, 0)$-block-encoding of $\frac{2}{\pi} H$, using $O\qty(\frac{d^{2k}}{\delta} \log^2 \frac{1}{\delta})$ samples of $\rho$. \\
\end{proof}
\noindent Next, we discuss how our reshuffling circuit in~\eqref{eq:reshuffling} can be used to estimate the full-singular value spectrum of unital quantum channels.  When the input is the Choi state of a unital channel, the circuit produces the following state; 
\begin{align}
    \rho_+ = \frac{1}{2d^2} |0\rangle \langle 0| \otimes I \otimes I + \frac{1}{2d^2}|0\rangle \langle 1| \otimes \mathcal{E}_A + \frac{1}{2d^2}|1\rangle \langle 0| \otimes \mathcal{E}_A^\dagger
    + \frac{1}{2d^2} |1\rangle \langle 1| \otimes I\otimes I. 
\end{align}
The eigenspectrum $\{ \lambda_i \}$ of this state coincides with $\{\frac{1}{2d^2}\pm \frac{1}{2d^2}\sigma_i\}$, where $\sigma_i$'s are the singular values of the quantum channel. By applying conventional spectrum estimation algorithms~\cite{PhysRevA.64.052311, odonnell2016efficientquantumtomographyii}, we see that the full-singular value spectrum can be estimated to $\epsilon$ precision in total variation distance with $O(d^6/\epsilon^{2})$ samples. Furthermore, since 
\begin{align}
    \left|\qty(\sum_{i=d^2+1}^{2d^2} \hat{\lambda}_i- \sum_{i=1}^{d^2} \hat{\lambda}_i)  - \frac{1}{d}\lVert\mathcal{R}(\mathcal{E}_B)\rVert_1 \right| = \left|\qty(\sum_{i=d^2+1}^{2d^2} \hat{\lambda}_i- \sum_{i=1}^{d^2} \hat{\lambda}_i)  - \qty(\sum_{i=d^2+1}^{2d^2} \lambda_i- \sum_{i=1}^{d^2} \lambda_i) \right|\leq \sum_{i=1}^{2d^2}|\hat{\lambda_i}-\lambda_i| \leq \epsilon, 
\end{align}
we can also estimate the 1st moment $\lVert\mathcal{R}(\mathcal{E}_B)\rVert_1$ up to precision $\epsilon$ using $O(d^6/\epsilon^{2})$ samples of $\rho_+$. \\

\section{Lower bound for the query complexity}\label{ap:samplecomplexity}
Let us define the unitary channel $\mathcal{U}(t)$ as 
\begin{align}
    \mathcal{U}(t)[\cdot] = U(t)[\cdot]U(t)^\dagger, \quad  U(t) = \mathrm{exp} \qty(-i\mqty[O & \mathcal{E}_A\\ \mathcal{E}_A^\dagger & O]t).
\end{align}
Now, suppose that there exists a protocol that universally creates a quantum channel $\widetilde{\mathcal{U}}(\delta, t)$ that is $\delta$-close in diamond norm to $\mathcal{U}(t)$, given the ability to apply the unknown channel $\mathcal{E}$ for $N$ times. 
\begin{align}
    \lVert\widetilde{\mathcal{U}}(\delta, t)-\mathcal{U}(t)\rVert_\diamond \leq \delta.
\end{align}
This process $\mathcal{P}$ can be seen as a quantum superchannel \cite{Chiribella_2008, Gour_2019} which takes $N$ copies of $\mathcal{E}$ as input and produces an output $\widetilde{\mathcal{U}}$. Most generally, $\mathcal{P}$ is a \textit{sequential superchannel}. Since $\mathcal{P}$ is applicable to arbitrary black-box channels, it can be used to solve the following channel discrimination task; Given the ability to apply a quantum channel that is promised to be either $\mathcal{E}_1$ or $\mathcal{E}_2$ for $N$ times, determine whether it is $\mathcal{E}_1$ or $\mathcal{E}_2$. Indeed, by producing  a single copy of $\widetilde{\mathcal{U}}_i$ by applying $\mathcal{P}$ to $\mathcal{E}_i$, the optimal success probability of discriminating between $\widetilde{\mathcal{U}}_1$ and $\widetilde{\mathcal{U}}_2$ is~\cite{Watrous_2018} 
\begin{align}
    \widetilde{P}^{\mathrm{suc.}}_N = \frac{1}{2}+\frac{1}{4}\lVert\widetilde{\mathcal{U}}_1(\delta, t)-\widetilde{\mathcal{U}}_2(\delta, t)\rVert_\diamond. 
\end{align}
In channel discrimination, the most general strategy allows feedback from the output, which is referred to as an \textit{adaptive} strategy. 
If we denote with $P^{\mathrm{suc.}}_N$ the success probability of discriminating between $\mathcal{E}_1$ and $\mathcal{E}_2$ with an adaptive strategy querying the channel $N$ times, 
\begin{align}
    \widetilde{P}^{suc.}_N \leq P^{\mathrm{suc.}}_N. 
\end{align}
Adaptive strategies can strictly outperform parallel strategies, so in general 
\begin{align}
\frac{1}{2}+\frac{1}{4}\lVert\mathcal{E}_1^{\otimes N}-\mathcal{E}_2^{\otimes N}\rVert_\diamond < P^{\mathrm{suc.}}_N. 
\end{align} 
However, it has been shown~\cite{PhysRevLett.118.100502} that the success probability for discriminating between two unknown channels with adaptive strategies reduces to the $N$-copy diamond distance and trace distance of Choi states for the special case of \textit{teleportation-covariant channels}; $\mathcal{E}(U\rho U^\dagger) = V\mathcal{E}(\rho)V^\dagger$, where $U$ is a teleportation unitary $U\in \mathbb{U}_d$ and $V$ is some unitary~\footnote{$\mathbb{U}_d$ consists of Heisenberg-Weyl operators, which is equivalent to Pauli operators for $n$ qubit systems}. For these channels, 
\begin{align}
\frac{1}{2}+\frac{1}{4}\lVert\widetilde{\mathcal{U}}_1(\delta, t)-\widetilde{\mathcal{U}}_2(\delta, t)\rVert_\diamond = \widetilde{P}^{\mathrm{suc.}}_N \leq P^{\mathrm{suc.}}_N = \frac{1}{2}+\frac{1}{4}\lVert\mathcal{E}_1^{\otimes N}-\mathcal{E}_2^{\otimes N}\rVert_\diamond =\frac{1}{2} +\frac{1}{4}\lVert\mathcal{E}_{1B}^{\otimes N}-\mathcal{E}_{2B}^{\otimes N}\rVert_1,
\end{align}
where we denote the Choi states of $\mathcal{E}_1$ and $\mathcal{E}_2$ by $\mathcal{E}_{1B}$ and $\mathcal{E}_{2B}$. It is important to emphasize that $N$ only depends on the parameters $\delta, t$, and is independent of the channel discrimination task. Thus, any lower bound for $N$ derived from a specific choice of the channel in the discrimination task applies generally to the original task of creating $\widetilde{\mathcal{U}}$ from $N$ copies of the black-box channel $\mathcal{E}$. \\

With this in mind, let us focus on the channel discrimination task of two teleportation-covariant channels $\mathcal{E}_1, \mathcal{E}_2$. First, let $f(\delta, t)$ denote the query complexity for obtaining $\widetilde{\mathcal{U}}(\delta, t)$ and let $t_\ast$ denote the minimum $t$ such that $\lVert\mathcal{U}_1(t)-\mathcal{U}_2(t)\rVert_\diamond=2$. We also denote with $\Theta(X)$ the smallest arc containing all the eigenvalues of the unitary $X$. Then, 
\begin{align}
     \lVert \mathcal{U}_1(t)-\mathcal{U}_2(t)\rVert_\diamond -\lVert \mathcal{U}_1(t)-\widetilde{\mathcal{U}}_1(\delta, t)\rVert_\diamond-\lVert \mathcal{U}_2(t)-\widetilde{\mathcal{U}}_2(\delta, t)\rVert_\diamond \leq \lVert\widetilde{\mathcal{U}}_1(\delta, t)-\widetilde{\mathcal{U}}_2(\delta, t)\rVert_\diamond \leq \lVert \mathcal{E}_{1B}^{\otimes {f(\delta, t)}}-\mathcal{E}_{2B}^{\otimes {f(\delta, t)}}\rVert_1.
\end{align}
which gives
\begin{align}
     2-\frac{2}{3} \leq  \lVert \mathcal{E}_{1B}^{\otimes f(\frac{1}{3}, t_\ast)}-\mathcal{E}_{2B}^{\otimes f(\frac{1}{3}, t_\ast)}\rVert_1 \leq 2\sqrt{1-F(\mathcal{E}_{1B}^{\otimes f(\frac{1}{3}, t_\ast)}, \mathcal{E}_{2B}^{\otimes f(\frac{1}{3}, t_\ast)}) } = 2\sqrt{1-F(\mathcal{E}_{1B}, \mathcal{E}_{2B})^{f(\frac{1}{3}, t_\ast)}}.
\end{align}
Here, $F(\rho, \sigma) = ||\sqrt{\rho}\sqrt{\sigma}||_1^2$ denotes the fidelity, and we have used the property $F(\rho_1 \otimes \rho_2, \sigma_1 \otimes \sigma_2)=F(\rho_1, \sigma_1)F(\rho_2, \sigma_2)$ as well as the Fuchs-van de Graaf inequality 
\begin{align}
    1-\sqrt{F(\rho, \sigma)} \leq \frac{1}{2}\lVert \rho-\sigma \rVert_1 \leq \sqrt{1-F(\rho, \sigma)}. 
\end{align}
If we denote $ H_1 :=\mqty[O & \mathcal{E}_{1A}\\ \mathcal{E}_{1A}^\dagger & O], H_2 := \mqty[O & \mathcal{E}_{2A}\\ \mathcal{E}_{2A}^\dagger & O]$, 
\begin{align}
    [H_1, H_2] = \mqty[\mathcal{E}_{1A}\mathcal{E}_{2A}^\dagger- \mathcal{E}_{2A}\mathcal{E}_{1A}^\dagger& O\\ O & \mathcal{E}_{1A}^\dagger\mathcal{E}_{2A}-\mathcal{E}_{2A}^\dagger\mathcal{E}_{1A}]. 
\end{align}
If $[H_1, H_2]=0$, we have 
\begin{align}
    \Theta(U_1^\dagger(t)U_2(t)) = 2d\lVert\mathcal{R}(\Delta \mathcal{E}_B)\rVert_\infty t \quad 
\end{align}
for sufficiently small $\Delta \mathcal{E}_B := \mathcal{E}_{1B}-\mathcal{E}_{2B}$. Then, we have 
\begin{align}
    1-F(\mathcal{E}_{1B}, \mathcal{E}_{2B})^{f\qty(\frac{1}{3}, t_\ast)} \geq \frac{4}{9}, \quad t_\ast=\frac{\pi}{2d\lVert\mathcal{R}(\Delta \mathcal{E}_B)\rVert_\infty}. 
\end{align}
To proceed further, let us specifically choose
\begin{align}
    \mathcal{E}_1(\rho)= \mathcal{E}_p(\rho) := \mathcal{D}_p \circ \Phi(\rho) , \quad \mathcal{E}_2(\rho)= \mathcal{E}_q(\rho) :=\mathcal{D}_{q} \circ \Phi(\rho), \quad (p=1, q=1-\epsilon)
\end{align}
where $\mathcal{D}_a(\rho) :=a\frac{I}{2}+(1-a)\rho, \, (a=p,q)$ denotes the qubit depolarizing channel and $\Phi(\rho)$ is defined as $\Phi(\rho) := \Tr_E (\rho) \otimes |\Phi_E\rangle \langle \Psi_E|$, where $\Tr_E$ denotes the partial trace over $n-1$ qubits and $|\Psi_E\rangle \langle \Psi_E|:=|00\cdots 0\rangle \langle 00\cdots 0|$ is an $n$-qubit pure state. Since the teleportation unitary $U \in \mathbb{U}_d$ consists of $n$ qubit Pauli operators for $d=2^n$, $\Phi$ becomes teleportation covariant. Thus, $\mathcal{E}_1, \mathcal{E}_2$ are also teleportation covariant channels. We can confirm that 
\begin{align}
    \mathcal{E}_{aB} &= \qty(a\frac{I_2}{2}\otimes \frac{I_2}{2}+(1-a)|\Phi_2^+\rangle \langle \Phi_2^+|)\otimes \qty(|\Psi_E\rangle \langle \Psi_E|\otimes \frac{I_{2^{n-1}}}{2^{n-1}}) \\
    \mathcal{R}(\mathcal{E}_{aB}) &= \qty(\frac{a}{2} |\Phi_2^+\rangle \langle \Phi_2^+|+ 2(1-a) \frac{I_2}{2}\otimes \frac{I_2}{2})\otimes \qty(\frac{1}{\sqrt{2^{n-1}}}
    (|\Psi_E\rangle \otimes |\Psi_E\rangle)\langle \Phi_{2^{n-1}}^+|),
\end{align}
so $\mathcal{E}_{1A} \mathcal{E}_{2A}^\dagger =\mathcal{E}_{2A} \mathcal{E}_{1A}^\dagger, \: \mathcal{E}_{1A}^\dagger \mathcal{E}_{2A} = \mathcal{E}_{2A}^\dagger \mathcal{E}_{1A}$. This implies $[H_1, H_2]=0$, and 
\begin{align}
\begin{dcases}
    \lVert \Delta \mathcal{E}_B \rVert_1 = \frac{3}{2}|p-q|\\
    d\lVert \mathcal{R}(\Delta \mathcal{E}_B)\rVert_\infty = \sqrt{\frac{d}{2}}|p-q|\\
    F(\mathcal{E}_{1B},\mathcal{E}_{2B}) = \qty[\sqrt{1-\frac{3}{4}p}\sqrt{1-\frac{3}{4}q}+\frac{\sqrt{p}}{2}\frac{\sqrt{q}}{2} \times 3]^2 = 1-\frac{3}{4}\epsilon^2+O(\epsilon^3)
\end{dcases}
\end{align}
Using the Bernoulli inequality $(1+x)^r \geq 1+rx$ for $x\geq -1$, and $r\in \mathbb{N}$, 
\begin{align}
    f(\frac{1}{3}, t_\ast) \geq \frac{16}{27\epsilon^2} = \frac{4}{3} \frac{1}{\lVert \Delta \mathcal{E}_B\rVert_1^2}=\frac{32}{27\pi^2}\qty(\sqrt{d}t_\ast)^2. 
\end{align}
Finally, if we set $m=\lceil \frac{1}{6\delta} \rceil$, we have $m\delta \leq \frac{1}{3}$, so for $t\geq t_\ast/m$, we have 
\begin{align}
    f(\delta, t) \geq \frac{1}{m} f(m\delta, mt) \geq \frac{1}{m} f(\frac{1}{3}, mt) \geq \frac{32}{27\pi^2} m(\sqrt{d}t)^2\geq \frac{16}{81 \pi^2} \cdot \frac{(\sqrt{d}t)^2}{\delta}. 
\end{align}

\section{1st moment via Fourier-cosine expansion}\label{ap:1stmoment}
Here we follow the methods employed in~\cite{Wei:2024znk} to directly estimate the 1st singular value moment, $S_1$. First, the Fourier-cosine expansion of the absolute function reads
\begin{align}
    |x| = \frac{\pi}{2}-\sum_{\ell=1}^{\infty} \frac{4}{\pi (2\ell-1)^2} \cos ((2\ell-1)x), \quad x \in [-1, 1].
\end{align}
Since $\lVert\mathcal{R}(\mathcal{E}_B)\rVert_\infty \leq 1$, (see Appendix~\ref{app:channelprop}, Lemma~\ref{lm:stateopnorm}), we can substitute the Hermitized matrix to obtain 
\begin{align}
    2\lVert\mathcal{R}(\mathcal{E}_B)\rVert_1 = \left\lVert\mqty[ O & \mathcal{R}(\mathcal{E}_B)\\ \mathcal{R}(\mathcal{E}_B)^\dagger & O]\right\rVert_1 = \frac{\pi}{2} (2d^2)-\sum_{\ell=1}^{\infty} \frac{4}{\pi (2\ell-1)^2} \Tr \qty[\cos \qty((2\ell-1)\mqty[ O & \mathcal{R}(\mathcal{E}_B)\\ \mathcal{R}(\mathcal{E}_B)^\dagger & O])].
\end{align}
Truncating the series expansion at level $L$ gives 
\begin{align}
    \left| \sum_{\ell = L+1}^\infty \frac{4}{\pi (2\ell-1)^2} \Tr \qty[\cos \qty((2\ell-1)\mqty[ O & \mathcal{R}(\mathcal{E}_B)\\ \mathcal{R}(\mathcal{E}_B)^\dagger & O])] \right | &< \sum_{\ell = L+1}^\infty \left| \frac{4}{\pi (2\ell-1)^2}  \right | \cdot (2d^2) \nonumber\\
    &< \frac{8d^2}{\pi} \int_{L}^\infty \mathrm{d} x \frac{1}{(2x-1)^2}  \nonumber\\
    &= \frac{4d^2}{\pi (2L-1)}
\end{align}
which can be bounded as $\frac{4d^2}{(2L-1)\pi} \leq \epsilon_1$ by choosing $L=\lceil \frac{2d^2}{\pi \epsilon_1} + \frac{1}{2}\rceil $. Next, consider a random variable $\mathbb{X}$ that takes the value $-\frac{\pi}{2}(2d^2), \frac{\pi}{2}(2d^2),0$ with the probability $\sum_{\ell=1}^{L} \frac{8}{\pi^2(2\ell-1)^2} p(0), \sum_{\ell=1}^{L} \frac{8}{\pi^2(2\ell-1)^2} p(1), 1-\sum_{\ell=L+1}^\infty \frac{8}{\pi^2(2\ell-1)^2}$, where
\begin{align}
\begin{dcases}
    p(0) = \frac{1}{2}+\frac{1}{4d^2} \Tr \qty[\cos \qty((2\ell-1)\mqty[ O & \mathcal{R}(\mathcal{E}_B)\\ \mathcal{R}(\mathcal{E}_B)^\dagger & O])]\nonumber\\
    p(1) = \frac{1}{2} -\frac{1}{4d^2} \Tr \qty[\cos \qty((2\ell-1)\mqty[ O & \mathcal{R}(\mathcal{E}_B)\\ \mathcal{R}(\mathcal{E}_B)^\dagger & O])]. 
\end{dcases}
\end{align}
Then, we obtain
\begin{align}
    \mathbb{E}[\mathbb{X}] = -\sum_{\ell=1}^L\frac{4}{\pi (2\ell-1)^2} \Tr \qty[\cos \qty((2\ell-1)\mqty[ O & \mathcal{R}(\mathcal{E}_B)\\ \mathcal{R}(\mathcal{E}_B)^\dagger & O])]. 
\end{align}
Here, $p(0), p(1)$ corresponds to the probability of obtaining the measurement outcomes $0,1$ in the DQC1 circuit with the ideal block-encoding unitary. Let $\widetilde{\mathbb{X}}$ denote the corresponding random variable when we replace the unitary with the approximate channel constructed by our block-encoding scheme. Defining the precision parameter
\begin{align}
    \epsilon_{2\ell-1}' = \frac{\pi(2\ell-1)}{2(2+\log(2L-1))} \frac{\epsilon_2}{2d^2}, 
\end{align}
one obtains 
\begin{align}
    |\mathbb{E}[\widetilde{\mathbb{X}}]-\mathbb{E}[\mathbb{X}]| &\leq \sum_{\ell =1}^L \frac{4(2d^2)}{\pi(2\ell-1)^2} \epsilon_{2\ell-1}'\nonumber\\
    &= \frac{2\epsilon_2}{2+\log (2L-1)} \sum_{\ell=1}^L \frac{1}{2\ell-1} \nonumber\\
    &< \frac{2\epsilon_2}{2+\log (2L-1)} \qty(1+\int_1^L \frac{1}{2x-1} \mathrm{d}x) = \epsilon_2
\end{align}
with $\widetilde{O}(d/\epsilon_{2\ell-1}')$ queries to the unknown channel, utilizing Theorem 1 in the main text. Computing the variance, one obtains
\begin{align}
    \mathrm{Var}[\widetilde{\mathbb{X}}] &= \mathbb{E}[\widetilde{\mathbb{X}}^2]-\mathbb{E}[\widetilde{\mathbb{X}}]^2 \nonumber\\
    & \leq |\mathbb{E}[\widetilde{\mathbb{X}}^2]-\frac{\pi^2}{4}(2d^2)^2|+|\mathbb{E}[\widetilde{\mathbb{X}}]^2-\frac{\pi^2}{4}(2d^2)^2|\nonumber\\
    &= \frac{\pi^2 (2d^2)^2}{4} \sum_{\ell=L+1}^\infty \frac{8}{\pi^2 (2\ell-1)^2}+|\mathbb{E}[\widetilde{\mathbb{X}}]-\frac{\pi}{2}(2d^2)||\mathbb{E}[\widetilde{\mathbb{X}}]+\frac{\pi}{2}(2d^2)| \nonumber\\
    &\leq \frac{\pi^2 (2d^2)^2}{4} \sum_{\ell=L+1}^\infty \frac{8}{\pi^2 (2\ell-1)^2}+|\mathbb{E}[\widetilde{\mathbb{X}}]-\frac{\pi}{2}(2d^2)|\qty(|\mathbb{E}[\widetilde{\mathbb{X}}]-\mathbb{E}[\mathbb{X}]|+|\mathbb{E}[\mathbb{X}]+\frac{\pi}{2}(2d^2)|) \nonumber\\
    &<\pi\epsilon_1 d^2+2\pi d^2(\epsilon_2+2\lVert\mathcal{R}(\mathcal{E}_B)\rVert_1+\epsilon_1) = O(d^2 \lVert\mathcal{R}(\mathcal{E}_B)\rVert_1), 
\end{align}
where we have used $|\mathbb{E}[\widetilde{\mathbb{X}}]|\leq \frac{\pi}{2}(2d^2)$, $|\mathbb{E}[\widetilde{\mathbb{X}}]-\mathbb{E}[\mathbb{X}]|< \epsilon_2$ and 
\begin{align}
     \left|\mathbb{E}[\mathbb{X}]+\frac{\pi}{2}(2d^2)\right| &=\left|2\lVert\mathcal{R}(\mathcal{E}_B)\rVert_1-\sum_{\ell = L+1}^\infty \frac{4}{\pi (2\ell-1)^2} \Tr \qty[\cos \qty((2\ell-1)\mqty[ O & \mathcal{R}(\mathcal{E}_B)\\ \mathcal{R}(\mathcal{E}_B)^\dagger & O])]\right|\nonumber\\
     & \leq 2\lVert\mathcal{R}(\mathcal{E}_B)\rVert_1 + \epsilon_1.
\end{align}
Employing the median of means estimation, it suffices to have $O\qty(\mathrm{Var}[\widetilde{\mathbb{X}}]\frac{\log \frac{1}{\delta} }{\epsilon_3^2})=O(d^2||\mathcal{R}(\mathcal{E}_B)||_1 \frac{\log \frac{1}{\delta} }{\epsilon_3^2})$ repetitions to produce an $\epsilon_3$-precise estimation with probability $P^{\textrm{suc.}}>1-\delta$. Thus, choosing $\epsilon_1=\epsilon_2=\epsilon_3=\frac{2\epsilon}{3}$, we can produce an $\epsilon$-close estimate for $||\mathcal{R}(\mathcal{E}_B)||_1$ with probability $P^{\textrm{suc.}}\geq 1-\delta$. The expectation value for the query complexity becomes 
\begin{align}\label{eq:1stmoment}
    O\qty(\sum_{\ell=1}^L  \frac{8}{\pi^2 (2\ell-1)^2}\frac{d(2\ell-1)^2}{\epsilon'_{2L-1}}\cdot d^2\lVert\mathcal{R}(\mathcal{E}_B)\rVert_1 \frac{\log \frac{1}{\delta} }{\epsilon_3^2}) &= O\qty(\sum_{\ell=1}^L  \frac{32(2+\log(2L-1))}{\pi^3(2\ell-1)} \frac{d^{5}}{\epsilon_2}\lVert\mathcal{R}(\mathcal{E}_B)\rVert_1 \frac{\log \frac{1}{\delta} }{\epsilon_3^2})\nonumber\\
    &= O\qty(\frac{16(2+\log(2L-1))^2}{\pi^3} \frac{d^{5}}{\epsilon_2}\lVert\mathcal{R}(\mathcal{E}_B)\rVert_1 \frac{\log \frac{1}{\delta} }{\epsilon_3^2})\nonumber\\
    &= \widetilde{O}\qty(\frac{d^{5} \lVert\mathcal{R}(\mathcal{E}_B)\rVert_1 \log \frac{1}{\delta}}{\epsilon^3}),
\end{align}
where $\widetilde{O}$ ignores logarithmic factors in $d$. Since $\lVert \mathcal{R}(\mathcal{E}_B)\rVert_1 \leq d$ (see Appendix~\ref{app:channelprop}, Lemma~\ref{lm:channeltrnorm}), the query complexity is $\widetilde{O}\qty(\frac{d^{6}\log \frac{1}{\delta}}{\epsilon^2})$ in the worst case. 

Let us now analyze the sample complexity of estimating the first moment via full-state tomography of Choi states. Consider a bipartite state $\rho \in \mathbb{C}^{d^2\times d^2}$. Let the spectral decomposition of $\rho_1-\rho_2$ be $\rho_1-\rho_2 = \sum_{i=1}^{d^2} \lambda_i \ket{i}\bra{i}$. From the triangle inequality of the Schatten $p$-norms, we obtain 
\begin{align}
    \left|\lVert\mathcal{R}(\rho_1)\rVert_p -\lVert\mathcal{R}(\rho_2)\rVert_p\right|  &\leq \lVert\mathcal{R}(\rho_1)-\mathcal{R}(\rho_2)\rVert_p\nonumber\\
    &\leq \sum_{i=1}^{d^2} |\lambda_i| \lVert\mathcal{R}(\ket{i}\bra{i})\rVert_p.
\end{align}
From this norm inequality and $\lVert\mathcal{R}(\rho)\rVert_1  \leq d\lVert\mathcal{R}(\rho)\rVert_2  = d\sqrt{\Tr \rho^2} \leq d$, we know that 
\begin{align}
    \left| \lVert \mathcal{R}(\rho_1)\rVert_1 -\lVert \mathcal{R}(\rho_2)\rVert_1\right|  \leq d\sum_{i=1}^{d^2} |\lambda_i| = d \lVert\rho_1- \rho_2\rVert_1.
\end{align}
To estimate $\lVert\rho_1-\rho_2\rVert_1$ up to precision $\epsilon$, it takes $O(d^4/\epsilon^{2})$ samples~\cite{Haah_2017}. Therefore, substituting $\epsilon \rightarrow \epsilon/d$, it takes $O(d^6/\epsilon^{2})$ samples to precisely estimate the 1st moment $\lVert\mathcal{R}(\rho)\rVert_1$ up to precision $\epsilon$ via full-state tomography in the worst case. Therefore, the approach in Eq. ~\eqref{eq:1stmoment} has comparable $d$-dependence. The advantage of this scheme is that $\lVert \mathcal{R}(\mathcal{E}_B \rVert_1$ is directly measured. This can be contrasted with the tomographic approach, where the learner is required to store the measurement data $\hat{\rho}$ in an exponentially large classical memory, realign the indices, and then perform diagonalization with a classical computer.

\section{Query complexity for learning singular value moments}\label{app::detailedproof} 
Suppose we are given access to the exact block-encoding of the Hermitized Liouville representation $H$, and utilized QSVT to obtain a matrix $M$ that is a $(1,5,\epsilon_F)$-block encoding of $P\qty(\frac{2}{\pi}H)$, where $P(x)$ is the approximation of the function $f(x)=\frac{1}{2}x^{r-1}|x|^{1+\alpha}$ up to precision $\epsilon''$, where $r+\alpha=q-2$, $r$ is a positive integer, and $\alpha \in (-1, 1)$, namely,  
\begin{align}
    |P(x) - \frac{1}{2}x^{r-1}|x|^{1+\alpha}| \leq \epsilon''. 
\end{align}
$P(x)$ has the same parity as $r-1$. Therefore, we can choose $P(x)$ and $f(x)$ to be an even function. Noting that 
\begin{align}
    \mqty[ \Sigma & O\\O & -\Sigma] = \frac{1}{2}\mqty[ U^\dagger & V^\dagger \\U^\dagger & -V^\dagger]\mqty[ O & U\Sigma V^\dagger\\V \Sigma U^\dagger & O]\mqty[ U & U\\V & -V], 
\end{align}
for any even function $f$, we see that 
\begin{align}
    f\qty(\mqty[ O & \mathcal{E}_A\\\mathcal{E}_A^\dagger & O]) = \frac{1}{2}\mqty[ U & U\\V & -V]\mqty[ f(\Sigma) & O\\O & f(-\Sigma)] \mqty[ U^\dagger & V^\dagger \\U^\dagger & -V^\dagger] = \mqty[Uf(\Sigma) U^\dagger & O\\ O & Vf(\Sigma)V^\dagger]. 
\end{align}
Now, we consider the circuit depicted below;  
\begin{figure}[H]
\centering
\includegraphics[width=0.5\columnwidth]{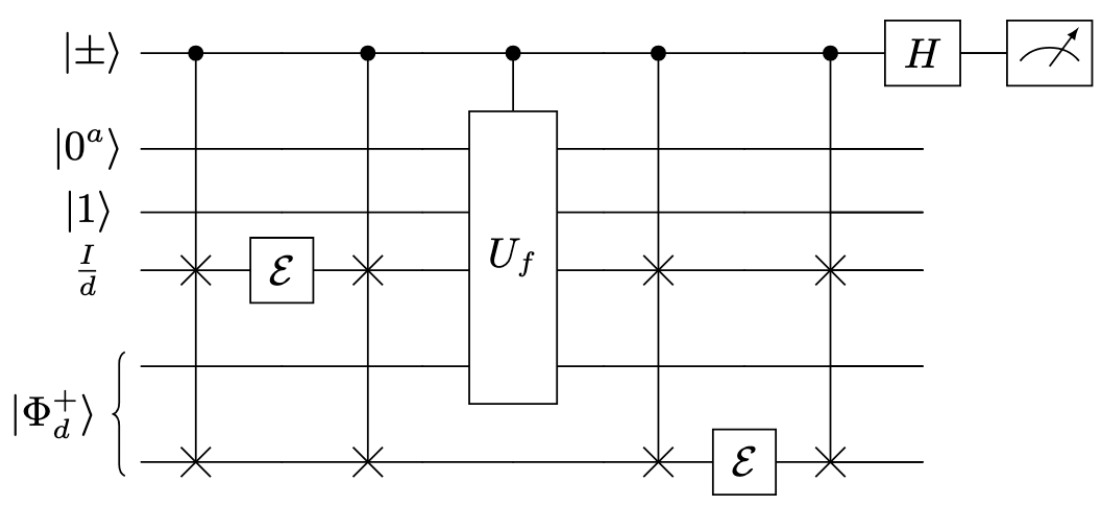}
\end{figure}
\noindent We can show that the probability distribution is 
\begin{align}
\begin{dcases}
    p(0) = \frac{1}{2}+\frac{1}{2}\mathrm{Re}\qty(\Tr [ \mathcal{E}_A^\dagger \mathcal{E}_A \langle 1|M|1\rangle])\\
    p(1) = \frac{1}{2}-\frac{1}{2}\mathrm{Re}\qty(\Tr [ \mathcal{E}_A^\dagger \mathcal{E}_A \langle 1|M|1\rangle]). 
\end{dcases}
\end{align}
Repeating the above circuit $O(\frac{\log (\frac{1}{\delta})}{\epsilon_H^2})$ times gives an estimate $\hat{S}'_q$ such that
\begin{align}
    \mathrm{Pr}(|\hat{S}'_q- \frac{1}{d^2}\Tr [ \mathcal{E}_A^\dagger \mathcal{E}_A \langle 1|M|1\rangle]| \leq \epsilon_H) \geq 1-\delta. 
\end{align}
Since $\Tr |AB| \leq \lVert A \rVert_1 \lVert B\rVert_\infty$,  $\lVert\frac{1}{d^2}\mathcal{E}_A^\dagger \mathcal{E}_A\rVert_1=\Tr \mathcal{E}_B^2 \leq 1$, $\lVert \langle 1|X|1\rangle\rVert_\infty  \leq \rVert X \lVert_\infty $, we have 
\begin{align}
    \frac{1}{d^2}\left| \Tr [\mathcal{E}_A^\dagger \mathcal{E}_A \langle 1|M|1\rangle] - \Tr [\mathcal{E}_A^\dagger \mathcal{E}_A \langle 1|P\qty(\frac{2}{\pi}H))|1\rangle] \right| &\leq \left\lVert M-P\qty(\frac{2}{\pi}H)\right\rVert_\infty \leq \epsilon_F.
\end{align}
We also have 
\begin{align}
    \Tr [\mathcal{E}_A^\dagger \mathcal{E}_A \bra{1} F\qty(\mqty[ O & \mathcal{E}_A\\\mathcal{E}_A^\dagger & O]) \ket{1}]=\Tr [V \Sigma^2 V^\dagger VF(\Sigma)V^\dagger] = \sum_i \sigma_i^2 F(\sigma_i)
\end{align}
for any even function $F$. Therefore, we obtain
\begin{align}
     \frac{1}{d^2}\left|\Tr [\mathcal{E}_A^\dagger \mathcal{E}_A \langle 1|P\qty(\frac{2}{\pi}H))|1\rangle]-\Tr [\mathcal{E}_A^\dagger \mathcal{E}_A \bra{1} f\qty(\mqty[ O & \mathcal{E}_A\\\mathcal{E}_A^\dagger & O]) \ket{1}] \right| 
     &\leq \frac{1}{d^2} \sum_{i=1}^{d^2} \sigma_i^2 \left|P\qty(\frac{\sigma_i}{\pi d^{1-k}})-f\qty(\frac{\sigma_i}{\pi d^{1-k}})\right| \nonumber\\
     &\leq \epsilon'', 
\end{align}
where we have used the fact $\frac{1}{d^2}\sum_i \sigma_i^2 = \Tr \mathcal{E}_B^2\leq 1$ (see Appendix~\ref{app:channelprop}, Lemma~\ref{lm:2ndmomentpurity}). We also have 
\begin{align}
    \frac{1}{d^2} \Tr \qty(\mathcal{E}_A^\dagger \mathcal{E}_A \bra{1} f\qty(\frac{2}{\pi}H)\ket{1}) = \frac{1}{d^2} \qty(\frac{1}{\pi d^{1-k}})^{q-2} \qty(\sum_{i=1}^{d^2}\sigma_i^q).
\end{align}
Utilizing the notion of the samplizer~\cite{wang_et_al:LIPIcs.ESA.2024.101}, we can obtain a quantum channel $\mathcal{U}$ that is $\delta'$-close in diamond norm to the unitary implemented in the above QSVT circuit $\mathcal{C}_{\textrm{QSVT}}$. Performing the Hadamard test with this approximated channel $\mathcal{U}$ produces the estimate $\hat{S}_q$ such that 
\begin{align}
    |\hat{S}_q-\hat{S}'_q| \leq \delta'
\end{align}
since $\Tr [M(\mathcal{C}_{\textrm{QSVT}}[\rho_0])-M(\mathcal{U}(\rho_0))] \leq ||M||_\infty ||\mathcal{C}[\rho_0]-\mathcal{U}(\rho_0)||_1 \leq ||\mathcal{C}_{\textrm{QSVT}}-\mathcal{U}||_\diamond =\delta'$ for any measurement operator $M$. Thus, we arrive at an $\epsilon_H + \epsilon_F+\epsilon''+\delta'$ estimate $\hat{S}_q$ for the quantity $\frac{1}{d^{2+(1-k)(q-2)}}\sum_{i=1}^{d^2} \sigma_i^q$ with 
\begin{align}
    O\qty(\frac{d^{2k+1} Q^2 \log^2 \qty(\frac{Q}{\delta'})\log \frac{1}{\delta} }{\delta' \cdot \epsilon_H^2 })
\end{align}
queries to the black-box channel. Substituting $\epsilon_H=\epsilon_F = \epsilon''=\delta'= \frac{\epsilon}{4}$ and $Q = O(1/\epsilon_F^{\frac{1}{q-2}})$, the overall query complexity for estimating $S_q$ to additive precision $\epsilon$ is $$O\qty(\frac{d^{2k+1} \log \frac{1}{\delta} \log^2 \qty(\frac{1}{\epsilon^{1+\frac{1}{q-1}}})}{\qty(\frac{d^q}{d^{2+(1-k)(q-2)}}\epsilon)^{3+\frac{2}{q-2}}})$$ with $k \leq \frac{1}{2}$ for general channels and $k\leq 1$ for unital channels. Therefore, we obtain the following query complexity; 
\begin{align}
\begin{dcases}
    \widetilde{O}\qty(\frac{\log \frac{1}{\delta}}{d^{\frac{3}{2}q-4}\epsilon^{3+\frac{2}{q-2}}}) \quad(\textrm{general channels})\\
    \widetilde{O}\qty(\frac{ \log \frac{1}{\delta}}{d^{3q-7}\epsilon^{3+\frac{2}{q-2}}}) \quad(\textrm{unital channels})\\
\end{dcases}
\end{align}

\section{\texttt{SWAP} circuit}\label{app:SWAP}
While in the main text we have discussed the general singular value moment $q>2, q\in \mathbb{R}$, when $q$ is an even integer, it is also possible to calculate the singular value moments with \texttt{SWAP} circuits implemented on many copies of the Choi state  $\mathcal{E}_B$. The \texttt{SWAP}s in the circuit are chosen so as to realize the following permutation pattern $J$ and $K$ in system 1 and 2, respectively:
\begin{equation}\label{eq:J}
    J = \mqty(1 & 2 & 3 & 4 & 5 & 6 & 7 & \cdots & 2k\\ 2k & 3 & 2 & 5 & 4 & 7 & 6 & \cdots & 1), \quad K = \mqty(1 & 2 & 3 & 4 & \cdots 2k-1 & 2k\\
    2 & 1 & 4 & 3 & \cdots 2k & 2k-1). 
\end{equation}
In a concrete form, this is expressed as 
\begin{align}
    \frac{1}{d^2}\sum_{i} \sigma_i^2 &= \Tr[\mathbb{F}_{(12)}\otimes \mathbb{F}_{(12)} \rho^{\otimes 2}] = \Tr \rho^2 \nonumber\\
    \frac{1}{d^4}\sum_{i} \sigma_i^4 &= \Tr[\mathbb{F}_{(4321)} \otimes \mathbb{F}_{(2143)} \rho^{\otimes 4}] \nonumber\\
    \frac{1}{d^6}\sum_{i} \sigma_i^6 &= \Tr[\mathbb{F}_{(632541)} \otimes \mathbb{F}_{(214365)} \rho^{\otimes 6}] \nonumber\\
    \frac{1}{d^8}\sum_{i} \sigma_i^8 &= \Tr[\mathbb{F}_{(83254761)} \otimes \mathbb{F}_{(21436587)} \rho^{\otimes 8}] 
\end{align}
\begin{figure}[H]
\centering
\includegraphics[width=0.5\columnwidth]{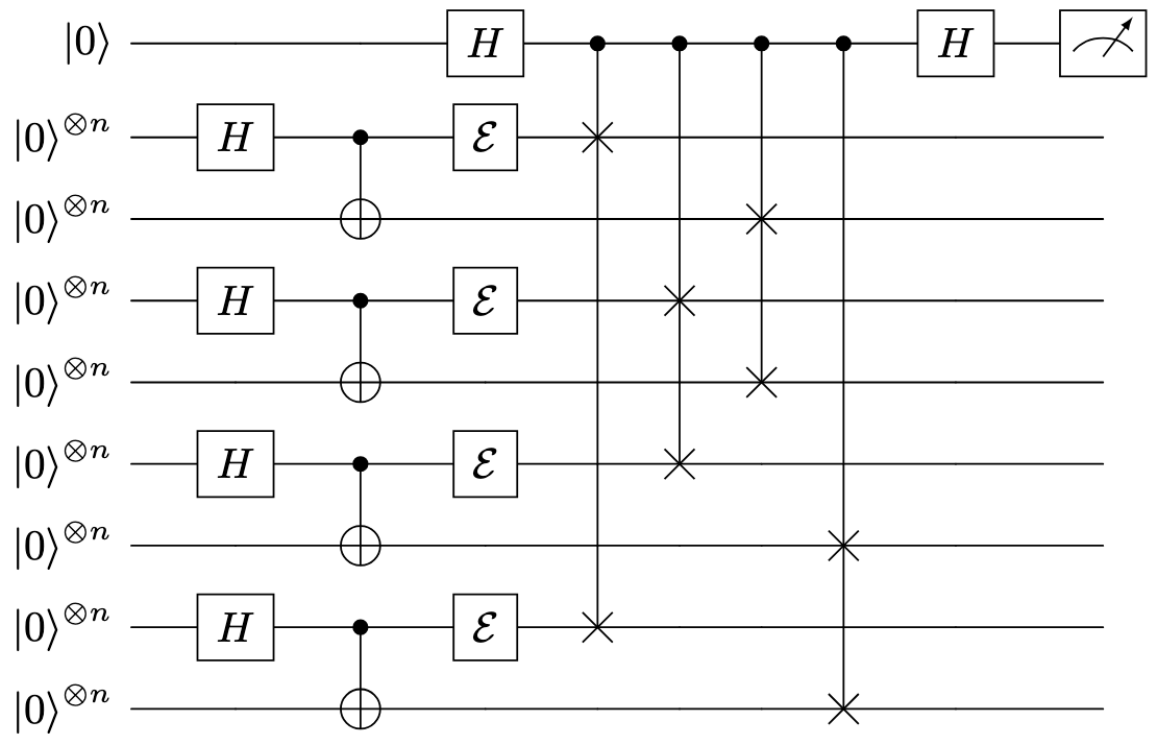}
\label{fig:swap_circuit.pdf}
\caption{A quantum circuit for calculating the $4$-th singular value moment.}
\end{figure}
Here, $\mathbb{F}_\sigma$ represents the \texttt{SWAP} operator that implements the permutation $(1, 2, \cdots, 2k)\rightarrow \sigma$. The query complexity for estimating $S_q = \frac{1}{d^q}\sum_i \sigma^q$ up to precision $\epsilon$ is $O\qty(\frac{\log \frac{1}{\delta}}{\epsilon^2})$. We note that the above circuit was originally discussed for general mixed states, rather than focusing on the Choi state of channels in the context of entanglement detection~\cite{PhysRevLett.101.190503}. While this scheme allows us to calculate $S_q$ for the special case where $q$ is an even integer, other moments such as odd integer or more general \textit{real moments} $q\in \mathbb{R}$ are unlikely to be calculated by the above strategy.

\end{document}